\pdfoutput=1
\documentclass[paper=letter, fontsize=10pt,leqno]{scrartcl}
\usepackage[body={5.35in,7.5in}]{geometry} 
\usepackage[english]{babel}				% English language/hyphenation	% Better typography
\usepackage{amsmath}				% Math packages
\usepackage{amsfonts}
\usepackage{amsthm}
\usepackage{hyperref}
 \usepackage{verbatim}
\usepackage{graphicx}				% Enable pdflatex
\usepackage{url}
\usepackage{amssymb}
\usepackage{bbm}
\usepackage{amscd}
\usepackage{dutchcal}
\usepackage[format=plain]{caption}
\usepackage{wrapfig}
 \usepackage{float}
\usepackage{sectsty}					% Custom sectioning (see below)
\allsectionsfont{\centering \normalfont\scshape}	% Change font of al section commands
\usepackage{mathtools}
\DeclareCaptionStyle{ruled}{labelfont=bf,labelsep=space}
\usepackage[]{forest}
\usepackage[displaymath, mathlines]{lineno}

\newtheorem{theorem}{Theorem}

\newtheorem{definition}[theorem]{Definition}
\newtheorem{proposition}[theorem]{Proposition}
\newtheorem{lemma}[theorem]{Lemma}
\newtheorem{corollary}[theorem]{Corollary}

\newtheorem*{proposition*}{Proposition}

\title{
		\vspace{-1in} 	
		\normalfont \normalsize \textsc{} \\ [0pt]
		\huge Entropy production in Knudsen thermodynamics of compartmented systems
}

\author{\normalfont \large 
  T. Chumley\footnote{Department of Mathematics and Statistics, Mount Holyoke College, 50 College St, South Hadley, MA 01075}, 
\ R. Feres\footnote{Department of Mathematics, Washington University, Campus Box 1146, St. Louis, MO 63130},
\ Z. Zhang\footnotemark[2]
}

\begin{document}
%\linenumbers

%\date{}

\maketitle

\begin{abstract}
\begin{center}
Abstract 
\end{center}

\medskip

{\small 
We investigate entropy production and nonequilibrium transport in a class of random dynamical systems modeling a
 Knudsen gas confined to a compartmented container. The system consists of  a single particle undergoing  random billiard motion, with collisions leading to random reflection or transmission through semi-reflecting, compartment-separating walls.
 
The stationary entropy production rate is first expressed  as the relative entropy between forward and time‑reversed path measures. Under a {\em reciprocity} assumption, used to introduce temperature into the general random billiard system, it is shown how this information-theoretic definition reduces to the classical thermodynamic formula: the mean energy transferred to the walls divided by their local temperatures\----a stochastic Clausius relation. We then develop a modular analysis of compartmented systems. Each open compartment is characterized by a compartment scattering operator and its sojourn statistics (expected time, collision number, and entropy produced). The sequence of compartment entrance states defines a Markov chain whose stationary distribution is used in a renewal‑reward theorem to assemble the compartment contributions into the global entropy production rate.

The framework is illustrated with a series of examples of increasing complexity governed by a generalized Maxwell-Smoluchowski
scattering operator and amenable to detailed and explicit analysis. Such operators are defined by a few parameters: temperature, a partial thermal accommodation,
the  height of potential barriers, and a porosity coefficient.
The central example is a cyclic three-compartment system consisting of two thermal walls at different temperatures and a potential barrier. For full thermal accommodation we obtain closed-form expressions for the entropy production rate and   net probability circulation around the cycle, revealing a thermal ratchet effect analogous to thermal transpiration. 
%The methods developed here provide a systematic route from microscopic boundary scattering laws to macroscopic entropy production in rarefied gases.
}
\end{abstract}

\section{Introduction}
In kinetic theory of gases, a {\em Knudsen gas} refers to a regime  where the gas is so rarefied that the molecular mean free path is much larger than the characteristic dimensions of the receptacle containing it. In such a regime,  transport properties are determined to a much greater extent by collisions 
between  the container's walls and  individual gas molecules than by intermolecular collisions. The assumption of local equilibrium\----traditionally made in continuous models of out-of-equilibrium thermodynamical systems\----is not valid for Knudsen gases. When different parts of the receptacle's walls
are kept at different temperatures,  the infrequent  intermolecular collisions are not sufficient to drive molecular velocity  to a near-equilibrium Maxwellian distribution in small regions. Hence the notion of a local temperature does not make sense. The lack of interaction between molecules, on the other hand,  makes it meaningful to reduce the study of such systems  to an investigation of the stochastic motion of a single  gas molecule. This is naturally done by employing a {\em random billiard} model, which we do here. Such a model has been widely used. We cannot cite enough papers on this topic to do it justice  but only mention \cite{ChumFer21,E2001,Feres2007,khanin-yarmola,CPSV2009, BruinKT2025, Rudzis2025,FZ2010,FZ2012}, which are a few of the   studies we are  familiar with. In what follows, we refer to the  gas molecule as the {\em particle}.

The present paper is mainly concerned with entropy production in the stationary regime for random billiard models in compartmented domains.  These are regions in space  bounded by   rigid walls and partitioned into compartments  by permeable (semi-reflecting) walls.    Upon reaching a wall,  the particle is assumed to either cross it into another compartment or reflect according to a choice of random scattering law. Reflection and transmission are specified by substochastic (transition) kernels whose sum equals $1$ and satisfy a reciprocity\----or time-reversibility\----condition relative to the Maxwellian distribution of velocities at temperature $T$. It is by means of this condition   that we assign a temperature $T$ to a wall. We note that while the mainstream framework of stochastic thermodynamics traditionally introduces environmental temperature via bulk fluctuation-dissipation relations in continuous Langevin dynamics or through transition rates in Markov jump processes \cite{seifert, qian, sekimoto}, such systems generally must introduce temperature through system parameters for bulk thermal baths. Drawing from classical kinetic theory, our approach defines wall temperatures locally via boundary scattering kernels that satisfy a reciprocity condition relative to a surface Maxwellian. While this is a standard approach in the Boltzmann equation literature, we believe this is relatively unique within the study of stochastic thermodynamics.

 As the particle undergoes  a random flight inside the compartmented container, it mediates the transfer of energy between different parts of the system,  thus promoting heat transfer within the system. Our main goal is to investigate, in the stationary regime of this stochastic process, the resulting  entropy production rate as a function of system parameters. Entropy production rate is  defined in terms of the relative entropy of the laws of finite-length forward and backward trajectories in the stationary regime. Our first main results, stated explicitly in Theorem \ref{second_law} and Corollary \ref{corollary_of_main_theorem} express the entropy production in terms of the stationary average energy flux at wall collisions divided by temperature at those collision points. This can be viewed as a stochastic version of Clausius' second law of thermodynamics: on average heat is transferred from hot source to cold sink. Such a formulation of entropy production establishes that in order to derive explicit values of entropy production it is necessary to understand the stationary distribution of the random billiard Markov chain.

 When illustrating the theory with concrete examples, we employ the following elaboration of the classical Maxwell-Smoluchowski  model of scattering. This is a parametric family of transition probability kernels depending on  parameters $p$, $C$, $\alpha$ and $T$,  where $p$  is the local fraction of the wall surface occupied by gaps or pores through which the particle can move across the wall; $C$ is a potential associated to these gaps that needs to be overcome for actual crossing; $\alpha$ is the Maxwell-Smoluchowski parameter (accommodation coefficient) representing the probability that a particle  undergoes diffuse reflection or transmission; more precisely, the particle's random reflected velocity is distributed according to the surface Maxwell-Boltzmann distribution $\mu_T$ at temperature $T$ (full thermalization), while $1-\alpha$ is the probability that a particle   reflects specularly or transmits without change. A precise analytic expression will be given shortly.  These wall parameters may be position dependent.
 
The second set of main results, stated   throughout Section \ref{sec:entropy-open-compartments}, establishes methods for the analysis of examples on compartmented domains. We begin with the study of entropy production and associated statistical quantities of a single compartment open system. Associated to each single compartment system is a transition probability kernel, the compartment operator, that encodes the probability distribution of outgoing states from the compartment for each given incoming state. By considering a wide class of closed compartmented systems in terms of their constituent open compartment subsystems, we then show in Theorem \ref{entropy_multicompartment} that the entropy production of the  compartmented system can be computed from the entropy production rates of its subsystems. Along the way we introduce an auxiliary Markov chain, the {\em entrance Markov chain}, which describes the sequence of states upon successive transitions between compartments. In addition to simplifying the expression of entropy production (Theorem \ref{entropy_multicompartment}), the entrance Markov chain is used to define the probability flux circulation in systems with multiple open compartments.
 
The work in this paper fits into the inter-related research literature on models of heat conduction and the mathematical analysis of nonequilibrium phenomena through entropy production rate. These lines of research are united by the theme of studying models that are simple enough to be amenable to detailed analytic study while also being rich enough to demonstrate the kinds of transport phenomena seen in real thermodynamic systems. While there has been a steady interest in mechanical models of heat conduction and the study of nonequilibrium behavior in systems with thermostatted boundaries, it is less common to find results that give analytic, functional relationships between thermodynamic quantities like entropy production rate and parameters of the system beyond thermostat temperatures. We note that in the case of anharmonic oscillators driven by stochastic heat baths, which under certain general assumptions can be modeled by stochastic differential equations, there is a fairly complete set of results. Techniques like those in \cite{qian} for stochastic differential equations are used to study entropy production rate in \cite{MR1889227}. However, such results are at the mesoscopic scale, and the picture at the microscopic scale, the category into which our models fall, is much less complete. There is a varied collection of microscopic models based on Hamiltonian mechanical systems coupled to stochastic heat baths \cite{lin, collet-eckmann, larralde, MR2200889, MR3217533}, but the main focus of study in such systems has mainly been related to characterizing stationary states and convergence to stationarity. For the study of entropy production rate, \cite{qian} gives a comprehensive study of countable state Markov chains. Entropy production for general classes of dynamical systems and Markov processes is studied in, for example, \cite{MR3517572, MR1705589, jpr, hongqian, seifert}, but this approach is a bit different than the one we take, where we are concerned with how  explicit, microscopic details of the model affect entropy production rate. We also note that entropy production for random billiard models is studied in \cite{ChumFer21} but for narrower classes of collision models and domains.

Central to this paper is the proposal that mathematical models of complex stochastic thermodynamic systems can be profitably formulated as multi-compartment systems, and that the integration of the open compartment units into the whole proceeds through a systematic method of analysis. One begins with a detailed analysis of the open compartment modules and then connects them through what we refer to as the {\em entrance Markov chain}. In this approach, quantities such as the entropy production rate can be obtained by first computing their expected values during sojourns in each compartment and then, using a renewal-reward argument together with the stationary distribution of the entrance Markov chain, combining these expected values to obtain the corresponding quantity for the whole system. This method is formalized for the entropy production rate in Theorem \ref{entropy_multicompartment}. The examples in Section \ref{sec:examples} serve primarily to illustrate this method of analysis and are therefore kept as simple as possible. Nevertheless, our main example\----the cyclic three-compartment system\----already demonstrates how such simple models can provide clear and explicit insight into physical phenomena such as thermal transpiration.

The rest of the paper is organized as follows. Section \ref{sec:definitions} introduces definitions and notation for random billiards on domains with compartments. Such systems are defined in terms of random scattering operators which are in turn used to define the Markov transition operator which governs the evolution of our random dynamical system. In this section we also introduce the generalized Maxwell-Smoluchowski model of surface scattering. In addition to serving as a model of thermalization, the generalized MS collision model  defines transport between compartments. Section \ref{sec:time-reversal-entropy-production} establishes basic facts about time reversal and entropy production of random billiard Markov chains. Section \ref{sec:introducing_temperature} introduces temperature through a notion of local detailed balance that we call reciprocity. In this section we connect  the information-theoretic formulation of entropy production introduced in Section \ref{sec:time-reversal-entropy-production} to the thermodynamic formulation of classical statistical physics which expresses entropy production in terms of energy flux and temperature. The remainder of the paper focuses on deriving explicit values of entropy production when surface scattering is governed by the MS collision model. Section \ref{sec:entropy-open-compartments} establishes general results computing entropy production and related quantities in compartment systems with the MS collision model. Section \ref{sec:examples} demonstrates the results of Section \ref{sec:entropy-open-compartments} through a series of explicit examples.

\section{General definitions}\label{sec:definitions}
\subsection{Domain with compartments}
Let $M_1, \dots, M_m$ be smooth  $n$-dimensional Riemannian manifolds with corners (as defined in \cite{lee}), whose boundaries are piecewise smooth, each
smooth piece being itself a   submanifold with corners.
We call these $M_j$ {\em compartments}. Each pair  of compartments, $M_j, M_k$,  may (or may not) {\em share} a smooth boundary piece in the following sense: they may 
contain    boundary pieces, $W$ and $W'$, respectively,  which are isometric relative to the restriction of the    Riemannian metrics from $M_j$ and $M_k$ so that,  by gluing the two compartments    along $W$ and $W'$
via this isometry, we obtain a smooth Riemannian metric on the union $M_j\cup M_k$. It may be the case that $M_j$ and $M_k$ share more than one pair of
boundary pieces. After gluing  the compartments    along all shared pairs of boundary pieces according to this scheme we obtain, by assumption, a compact, smooth, oriented, $n$-dimensional,  Riemannian manifold $M$ with corners and piecewise smooth boundary, together with  a collection of smooth, codimension-one submanifolds resulting from the shared pair identifications. These smooth submanifolds will be called {\em separating walls}, while the boundary pieces forming the boundary of $M$ (those which were not paired) will be called {\em boundary walls}. We denote by  $\mathcal{W}$  the union of all the walls and use  $W$ when referring to a  single wall. The 
union of the boundary walls may also be written $\partial M$.

 We think of $M$ as a collection of boxes glued together along some of their faces.   It will represent the configuration manifold of a  mechanical system, such as a molecule, which can be,  at any given time, inside one of those boxes, and can move into other boxes upon colliding with separating walls according to a probabilistic  reflection/transmission law to be introduced shortly.  
 
   The reader who   prefers to think of $M$ as a submanifold of Euclidean coordinate space $\mathbb{R}^n$  will not incur a great loss; indeed,
  the examples we give
 will be in Euclidean space or flat torus. 
The main reason for the more general Riemannian setting is to allow $M$ to contain rotational degrees of freedom, although systems incorporating angular momentum are not explicitly considered here.  

To each wall $W$ we assign an orientation defined by a smooth, unit length  vector field $\mathbf{n}$,  orthogonal to the tangent space $T_qW$ at each regular point $q$ (namely, a point at which  the tangent space is uniquely defined).   On separating walls, we select  $\mathbf{n}$  arbitrarily from the two possible choices of sign 
and, on boundary walls, we choose   $\mathbf{n}$ pointing  into $M$. (This choice is, clearly, arbitrary. In some examples later on, we may use
the opposite convention. In many places, superscripts $\text{i}$ and $\text{o}$, for {\em incoming} and {\em outgoing}, will resolve any possible ambiguity. See further  below.)
  The value of $\mathbf{n}$ at $q\in W$ may be written $\mathbf{n}_q$ or $\mathbf{n}(q)$, depending on notational convenience. This applies to  vector fields more generally.

  Elements of the tangent bundle $TM$ will be referred to as {\em states} of the mechanical system. We write  $x=(q,v)\in TM$,  $v\in T_qM$, and the base-point projection  is denoted by
$\pi:TM\rightarrow M$. By definition,   $\pi(q,v)=q$. Given a point $q$  on a wall $W$, we write
%nonumber
\begin{equation}\label{def:W_q}
W_q:=\{v\in T_qM: \langle \mathbf{n}_q,v\rangle=0\},\ \ 
 W_q^\pm:=\{v\in T_qM: \pm\langle \mathbf{n}_q,v\rangle> 0\}.
\end{equation}
Thus $W_q=T_qW$  and the $W_q^\pm$ are the two half-spaces in $T_qM$. If $W$ is a boundary wall, $W_q^-$ (respectively, $W_q^+$) is the subspace of {\em pre-collision} (respectively, {\em post-collision}) velocities. For separating walls, a particle  (represented by a time-parametrized curve  in $M$ defining  the evolving  configuration of the mechanical system) may impinge on $W$, or be emitted from $W$, at $q$ with velocity 
in either  $W_q^-$ or $W_q^+$. Later on, the superscripts $\text{`i'}$ and $\text{`o'}$ will be used when referring  to {\em incoming} and {\em outgoing} states  at a collision point.   For example, a measure representing the probability distribution of incoming velocities in $W_q^-$, hence approaching $q$ from the side of $W$ to which $\mathbf{n}_q$ is pointing, may be written $\nu_q^{-\text{i}}$. 
 Let us further  denote by $\mathcal{N}$ the restriction of $TM$ to $\mathcal{W}$.
 This is a vector bundle over $\mathcal{W}$ with base-point map $\pi:\mathcal{N}\rightarrow\mathcal{W}$ and fiber $T_qM$ over $q\in \mathcal{W}$; and let $\mathcal{N}^\text{o}\subseteq \mathcal{N}$ be the fiber bundle such that $\mathcal{N}_q^{\text{o}}=T_qM$ whenever
 $q$ lies in a separating wall, and $\mathcal{N}_q^{\text{o}}=W_q^+$ when $q$ lies on a boundary wall.  We similarly define $\mathcal{N}^{\text{i}}$ with  $\mathcal{N}_q^{\text{i}}=W_q^-$  for $q$ in a boundary wall.

The Riemannian volume measure on $\mathcal{W}$ will be denoted $A$. This is the Lebesgue area measure
for surfaces in $\mathbb{R}^3$. Subscripts may be appended when referring to the restriction of $A$ to subsets of individual walls. 

\subsection{Free-flight segments}
We suppose that the Riemannian metric  on $M$ is complete and denote by
 $t\mapsto \varphi_t(x)$  the {\em geodesic flow} in $TM$. If $x=(q,v)$ and $q$ lies in a compartment $M_j$, a {\em free-flight segment} is a segment of geodesic
between consecutive  collisions with  walls. Recall that the geodesic flow is the Hamiltonian flow on $TM$ for the free motion
of a particle with mass parameter $m$ and kinetic energy function 
%keep number
\begin{equation}\label{def:kinetic} 
E(q,v)=\frac12m|v|_q^2,
\end{equation}
where $|v|_q^2=\langle v,v\rangle_q$.   If the  particle represents the configuration of an extended system, $m$ is the total mass, with other constituent masses and moments of inertia incorporated into the Riemannian inner product.  For
$x=(q,v)$ with $q\in \mathcal{W}$ and $v$   in $T_qM$, let
\begin{equation*} 
t(x):=\inf\left\{t>0: \pi(\varphi_t(x))\in \mathcal{W}\right\}.
\end{equation*}
If $q$ belongs to a boundary wall $W$, it is assumed that $v\in W_q^+$. 

\begin{definition}[Return map]
The {\em return (to a wall) map} $\mathcal{T}:\mathcal{N}^{\text{o}}\rightarrow \mathcal{N}^{\text{i}}$ is  defined by 
%keep number
\begin{equation}\label{return_W}
\mathcal{T}(x):=\varphi_{t(x)}(x).
\end{equation}
\end{definition}
Upon reaching a wall, the particle trajectory either reflects off it according to some random distribution of velocities or, in the case of a separating wall, moves 
into an adjacent compartment. The process of reflection or transmission will be described shortly in terms of Markov transition operators. 
An {\em orbit} of the random motion consists of a sequence of free-flight segments separated by the two kinds of wall collision events: reflection or transmission.

If we further introduce  a smooth potential  function  $\Phi:M\rightarrow \mathbb{R}$,  the geodesic flow is replaced with  the Hamiltonian flow 
 associated to the energy function 
 \begin{equation*}
 E = E_0 + \Phi\circ \pi
 \end{equation*}
 where $E_0$ is the kinetic energy function previously defined in (\ref{def:kinetic}). The presence of a non-constant potential function modifies the flight segments but does not affect
 the collision events at walls. If a potential $\Phi$ is present, we further assume that $\mathcal{T}$ is   finite at every regular point in $\mathcal{W}$.

 \subsection{The random dynamical system}\label{rds}
  The space of Borel probability measures on a topological space will be denoted $\mathcal{P}(X)$. Two probability measures are {\em equivalent} if they are mutually
 absolutely continuous. If $\mu\in \mathcal{P}(X)$ and $f$ is a $\mu$-integrable function on $X$, we denote as before the $\mu$-integral of $f$ by
 $\mu(f):=\int_Xf(x)\, d\mu(x)$. A family of  measures $\mu_s$ indexed by the elements $s$ of another measurable  space is said to depend measurably on $s$ if for all bounded continuous $f$ the
 function $s\mapsto \mu_s(f)$ is measurable. When writing measure elements in integration the two notations $d\mu(x)=\mu(dx)$ will be used according to convenience.  Recall that if $F:X\rightarrow Y$ is a measurable map between measurable spaces and $\mu$ is a probability measure on $X$, the {\em push-forward}
 of $\mu$ under $F$ is the probability measure $F_*\mu$ on $Y$ defined by
 \begin{equation*}
 (F_*\mu)(U):=\mu(F^{-1}(U))
 \end{equation*}
 on any measurable subset $U$ of $Y$. 
 From this definition it follows that if $f$ is a bounded continuous (test) function on $Y$, then
 $$\int_Y f\, (F_*\mu) =\int_X f\circ F\, d\mu. $$
 If $F$ is a self-map of $X$, then $\mu$ is said to be {\em invariant} under $F$ if $F_*\mu=\mu$.

For the next definition, recall the return to $\mathcal{W}$ map $\mathcal{T}$ defined in Equation (\ref{return_W}).
 \begin{definition}[General random scattering operator] 
Let $W\subseteq \mathcal{W}$ be a wall, and $q\in W$ a regular point. We define a random {\em scattering operator}   at $q$ to be  
a map that associates to  $v\in T_qM$ a probability measure \begin{equation*}S_{x}\in \mathcal{P}(T_qM),\end{equation*}
assumed  measurable in $x=(q,v)$.  Thus if $f$ is a bounded  continuous function on $TM$,  
\begin{equation*}
x\mapsto S(f)(x):=S_x(f)
\end{equation*}
is measurable. The conditional probability notation
\begin{equation*}
S_q(f|v)=S(f)(q,v)=\int_{T_qM} f(u) S_q(du|v)
\end{equation*}
will also be used.
 If $W$ is a boundary wall, then $S$ is a Markov kernel that to  each $v\in W_q^-$ (necessarily an incoming vector) associates a probability
measure $S_x\in \mathcal{P}(W_q^+)$. If $W$ is a separating wall, then $v$ may lie in either $W_q^-$ or $W_q^+$ and $S_x$ is a probability measure
on the whole of $T_qM$.  
 \end{definition}
 
 The just  introduced $S$  gives rise at each $q$ to a self-map $S_q$ on  $\mathcal{P}(T_qM)$ defined by $\nu\mapsto \nu S_q$ such that
 \begin{equation}\label{nu_S}
 (\nu S_q)(f) = \int_{T_qM} S(f)(q,v)\, d\nu(v),
 \end{equation}
  for each bounded
continuous $f:T_qM\rightarrow \mathbb{R}$. Going forward, we refer to such bounded continuous functions
used to characterize a measure  as {\em test functions}.

\begin{definition}[The one-step random map]\label{one_step_generator}
The map $B:\mathcal{N}^{\text{\em o}}\rightarrow \mathcal{P}(\mathcal{N}^{\text{\em o}})$ defined by 
\begin{equation*}
B(x)=B_x:= S_{\mathcal{T}(x)},
\end{equation*}
where $S_{(q,v)}:=S_q(\cdot|v)$ is the random scattering operator, will be called the {\em one-step random map}. Thus $B(x)$ is 
a probability measure on $\mathcal{N}_{\pi(\mathcal{T}(x))}^{\text{\em o}}$.
Given a probability measure $\nu$ on $\mathcal{N}^\text{\em o}$,   the action of $B$ on $\nu$, denoted $\nu B$, is another probability measure on $\mathcal{N}^\text{\em o}$ defined by
\begin{equation*}
(\nu B)(\cdot) :=\int_{\mathcal{N}^{\text{\em o}}}  B_x (\cdot)\, d\nu(x).
\end{equation*}
 If $f$ is a test function on $\mathcal{N}^{\text{\em o}}$, then the integral of $f$ 
with respect to $B_x$ may be expressed in any one of the following notations:
\begin{equation*}
(Bf)(x) = B_x(f) =(\delta_x B)(f) = S_{\mathcal{T}(x)}(f). 
\end{equation*}
Here $\delta_x$ is the point  mass (or Dirac delta-measure) supported on $x$. 
\end{definition}
The map   on probability measures defined by
 the one-step random map
is the composition of two operations: push-forward by the free-flight map $\mathcal{T}$
followed by the action of the random scattering kernel $S$: 
\begin{equation*}
\nu B = (\mathcal{T}_*\nu) S = \nu \mathcal{T}^* S.
\end{equation*}
Here $\mathcal{T}^*$ is the pull-back of functions: $\mathcal{T}^*f:=f\circ \mathcal{T}$. Thus 
\begin{equation*}
B=\mathcal{T}^*\circ S. 
\end{equation*}
We may write the pull-back operator $\mathcal{T}^*$  (for $\mathcal{T}$ and other maps) simply as $\mathcal{T}f=f\circ\mathcal{T}$ and reserve the superscript $*$ for Hilbert space adjoint. With this in mind, the push-forward of a measure  $\nu$ 
may, occasionally, be written  $\mathcal{T}_*\nu = \nu\mathcal{T}$, in keeping with the convention that operators will act on measures from the right and on functions from the left.

The operator  $B$ generates by iteration a random dynamical system under the composition of random maps:
\begin{equation*}
(B\circ B)_x(\cdot) = \int_{\mathcal{N}^\text{o}} B_y(\cdot) dB_x(y).
\end{equation*}

We further define
\begin{equation*}
\mathcal{D}^\text{o}:=\left\{(x,y)\in \mathcal{N}^{\text{o}}\times \mathcal{N}^\text{o}: \pi(y)=\pi(\mathcal{T}(x))\right\}.
\end{equation*} 
A doubly infinite  orbit of the random dynamical system is a sequence $\dots, x_{-1}, x_0, x_1, \dots$ in which every pair $(x_i, x_{i+1})$ lies in $\mathcal{D}^\text{o}$ and,
for each $x_i$, the state  $x_{i+1}$ is sampled from the probability measure $B_{x_{i}}$. More often, we consider semi-infinite orbits, $x_0, x_1, \dots$, similarly defined. Given an initial state $X_0$ (a random variable with a given initial probability distribution), an orbit of the random dynamical system can also be represented by a sequence of random variables $X_0, X_1, X_2, \dots$
on the state space $\mathcal{N}^{\text{o}}$ such that the  $X_i$ are related to the one-step map by the equation
\begin{equation*}
(Bf)(x)=\mathbb{E}[f(X_{i+1})|X_i=x],
\end{equation*}
 where $\mathbb{E}[\cdot|\cdot]$ indicates conditional expectation. Thus the sequence $X_0, X_1,\dots$ constitutes a discrete time Markov chain 
 with state space $\mathcal{N}^{\text{o}}$ and transition operator $B$.

 \begin{definition}[Stationary measure]
 A probability measure $\nu$ on $\mathcal{N}^{\text{\em o}}$ is said to be stationary for the Markov chain generated by $B$ if $\nu=\nu B$. Equivalently,
 for every test function $f$, 
 \begin{equation*}
 \int_{\mathcal{N}^{\text{\em o}}} S(f)(\mathcal{T}(x))\, d\nu(x) = \int_{\mathcal{N}^{\text{\em o}}} f(x)\, d\nu(x).
 \end{equation*}
 If $\nu$ is a stationary measure, we may also  write $\nu^{\text{\em i}}:=\mathcal{T}_*\nu$ and $\nu^{\text{\em o}}:=\nu$.
 \end{definition}
 
 Let $\mathcal{H}=L^2(\mathcal{N}^\text{o},\nu)$, where the probability measure $\nu$ is stationary for $B$, and denote by $\langle\cdot,\cdot\rangle$ the Hilbert space inner product.
\begin{proposition}
The one-step random map $B$, as an operator in $\mathcal{H}$,  has norm $\|B\|_2=1$. 
\end{proposition}
\begin{proof}
Due to the Cauchy-Schwarz inequality and   $\nu$ and $S_x$ being probability measures, we have
$$\left|S_{\mathcal{T}(x)} f\right|^2 =\left|\int_{\mathcal{N}^\text{o}_{\pi(\mathcal{T}(x)})} f(y) S\left(dy|\mathcal{T}(x)\right)\right|^2\leq 
\int_{\mathcal{N}^{\text{o}}_{\pi(\mathcal{T}(x)})} |f(y)|^2 S\left(dy|\mathcal{T}(x)\right).$$
Therefore,
$$ \|Bf\|^2_2=\int_{\mathcal{N}^\text{o}} |S_{\mathcal{T}(x)}f|^2 \, d\nu(x)\leq \int_{\mathcal{N}^\text{o}} \int_{\mathcal{N}^{\text{o}}_{\pi(\mathcal{T}(x)})} |f(y)|^2 \, S\left(dy|{\mathcal{T}(x)}\right)\, d\nu(x)  =(\nu B)\left(|f|^2\right),$$
for any $f\in \mathcal{H}$. Stationarity of $\nu$ implies $(\nu B)(|f|^2)=\nu\left(|f|^2\right)=\|f\|^2_2$, hence $\| B\|_2\leq 1$. Furthermore, 
as $B1=1$ (here $1$ is the constant function identically equal to $1$),  we conclude that $\|B\|_2=1$. 
\end{proof}
\subsection{The  generalized  MS collision model}
 In the examples to be discussed later,
 we employ  a particularly simple and useful mathematical model of surface scattering,  similar to the well-known  Maxwell-Smoluchowski (MS) model of  reflection off a rigid   surface (to be recalled shortly), but modified so as to   allow transmission. In essence, we imagine that, at a small scale, the separating wall contains small 
 pores with an assigned potential value; when the particle collides with  the wall, it has a certain probability of impinging on a pore and, if its kinetic energy
 is sufficient to overcome the potential barrier, the particle can move to the other side. If the particle does not impinge on a pore or,  if it does, but without sufficient
 energy to overcome the potential barrier, it reflects according to the  standard MS model. 
  The parameters characterizing the model at a point $q$  are
 \begin{equation*}
 0\leq \alpha_r, \alpha_t \leq 1, \ \ 0\leq p\leq 1, \ \ C\geq 0, \ \ T>0,
 \end{equation*}
 which are interpreted as: (1) the probabilities  $\alpha_r, \alpha_t$ of diffuse reflection or transmission at temperature $T$, known as   {\em (thermal) accommodation coefficients},  (2) the fraction $p$ of surface area occupied by pores or, more precisely,  the probability $p$ that the particle 
 impinges on a  pore,  (3) the pore  potential $C\geq 0$.

It will be apparent as we describe the  model (that is, the random scattering operator $S$) in detail that it admits many possible variations, corresponding to somewhat different physical interpretations, but we adopt the following for concreteness.    Let the mass of the particle be $m$ and  introduce the speed parameter \begin{equation*}\gamma:=\sqrt{2C/m}.\end{equation*} 
For $\epsilon\in\{-,+\}$, define the {\em reflection} and {\em transmission} subsets of $W_q^\epsilon$ (see Figure \ref{RT}) by
\begin{equation*}
\mathsf{R}_q^\epsilon:=\{v\in W_q^\epsilon: |\langle v,  \mathbf{n}_q\rangle|<\gamma\}, \ \ 
\mathsf{T}_q^\epsilon:=\{v\in W_q^\epsilon:  |\langle v,  \mathbf{n}_q\rangle| >\gamma\}.
\end{equation*}
Thus the velocity vector $v$ belongs to $\mathsf{R}_q^\epsilon$ if the component of $v\in W_q^\epsilon$ perpendicular to the wall $W$ at $q$ does not
have sufficient kinetic energy to overcome the potential barrier:
\begin{equation*}
\frac12 m  \langle v,  \mathbf{n}_q\rangle^2<C.
\end{equation*}
The plane consisting of $v$ such that  $ \langle v,  \mathbf{n}_q\rangle=\epsilon \gamma$ partitions $W_q^\epsilon$ into the reflection and transmission subsets. We
indicate by 
\begin{equation*}\label{def_ref}
\overline{v}:=v-2\langle v,\mathbf{n}_q\rangle \mathbf{n}_q
\end{equation*}
the specular reflection  of $v$ on the tangent plane to $W$ at $q$. 
Let $\mu_T$ be the surface Maxwellian at temperature $T$. See Definition \ref{maxwellian_definition}.

   \begin{figure}[ht]
\begin{center}
\includegraphics[width=3in]{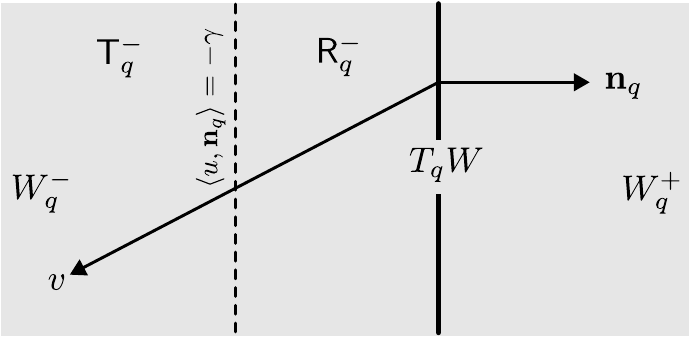}\ \ 
\caption{{\small  Partition of $W_q^-$ into reflection and transmission regions (separated by the dashed line). A particle with velocity $v$ in
$\mathsf{T}_q^-$ crosses the wall with unchanged velocity (thus remaining in $W_q^-$)   with probability $p$.}}
\label{RT}
\end{center}
\end{figure}

We only define the blocks $S_q^{--}$ and $S_q^{-+}$, the definitions of $S_q^{++}$ and $S_q^{+-}$ being similar.
Given a pre-collision velocity $v\in W_q^-$ and a Borel set $U\subseteq W_q^-$, the probability that the particle transmits and  
the transmitted velocity lies in $U$, conditional on the incoming velocity being $v$, is
\begin{equation*}
S_q^{--}(U|v):=p \mathbbm{1}_{\mathsf{T}^-_q}(v)\left[\alpha_t\mu_T(U)+(1-\alpha_t)\delta_{v}(U)\right].
\end{equation*}
If $U\subseteq W_q^+$, the probability of reflection, with reflection velocity in $U$, is
\begin{equation*}
S_q^{-+}(U|v):=\left(1-p\mathbbm{1}_{\mathsf{T}^-_q}(v)\right)\left[\alpha_r\mu_T(U)+(1-\alpha_r)\delta_{\overline{v}}(U)\right].
\end{equation*}
Note that 
\begin{equation*}
S_q^{--}(T_qM|v) + S_q^{-+}(T_qM|v)=1.
\end{equation*}

The {\em accommodation coefficients} $\alpha_r$ and $\alpha_t$ may be interpreted as defining the efficiency with which the
wall (connected to an unspecified heat bath) can exchange heat with the particle. 

A few special cases for $S$ may be useful to keep in mind:
\begin{enumerate}
\item {Boundary wall} (the standard Maxwell-Smoluchowski model): $p=0, C=0$.
\item {Potential barrier} (no particle-wall accommodation): $\alpha_r=\alpha_t=0$, $p=1$, $C>0$.
\item {Sieve} (no particle-wall accommodation in transmission): $p>0$, $C=0$, $\alpha_t=0$.
\end{enumerate}

\section{Time reversal and entropy production}\label{sec:time-reversal-entropy-production}
\subsection{Generator of the time reversed Markov chain}
\begin{proposition}\label{B*-invariant}
Let $B^*$ be the Hilbert space adjoint of $B$ in $\mathcal{H}$. Then the $B$-invariant measure $\nu$ is also $B^*$-invariant.
\end{proposition}
\begin{proof}
Simply note that
$$ \nu B^*f =\langle B^*f, 1\rangle = \langle f, B1\rangle =  \langle f,1\rangle=\nu(f),$$
for all test functions $f$. Therefore $\nu B^*=\nu$ as claimed.
\end{proof}

\begin{proposition}
Define the involution $\mathcal{J}:=J\mathcal{T}:\mathcal{N}^\text{\em o}\rightarrow \mathcal{N}^\text{\em o}$, where $Jv=-v$ is the velocity {\em flip map}, and suppose
that $\mathcal{J}_*\nu$ is absolutely continuous with respect to $\nu$, so that the Radon-Nikodym  derivative 
$$\rho_\nu:=\frac{d(\mathcal{J}_*\nu)}{d\nu}$$ is defined.  Then 
$$ \rho_\nu(x) \rho_\nu\left(\mathcal{J}(x)\right)=1$$
for $x\in \mathcal{N}^{\text{\em o}}$
and $\nu$ is also absolutely continuous with respect to $\mathcal{J}_*\nu$ with Radon-Nikodym derivative
$$\frac{d\nu}{d(\mathcal{J}_*\nu)}(x) = \rho_\nu\left( \mathcal{J}(x)\right).$$
Furthermore, the adjoint of $\mathcal{J}$ is
$$\mathcal{J}^*=\mathcal{M}_{\rho_\nu}\mathcal{J},$$
in which $\mathcal{M}_{\rho_\nu}$ denotes the multiplication operator by $\rho_\nu$.
\end{proposition}
\begin{proof}
That $\mathcal{J}$ is an involution is due to   $\mathcal{T}$ being time reversible.  This  property of $\mathcal{T}$ is inherited  from time reversibility of the geodesic flow, or the Hamiltonian system with  potential function $\Phi$. 
As a general fact, if a measure $\nu_1$ is absolutely continuous with respect to $\nu_2$ and $\mathcal{J}$ is a measurable transformation,
then $\mathcal{J}_*\nu_1$ is absolutely continuous with respect to $\mathcal{J}_*\nu_2$. In particular, if $\mathcal{J}$ is an involution and $\mathcal{J}_*\nu$ is absolutely continuous with respect to $\nu$, then $\nu=\mathcal{J}_*^2\nu$ is absolutely continuous with respect to $\mathcal{J}_*\nu$. Furthermore,
\begin{align*}\int_{\mathcal{N}^\text{o}}f\, d\nu &= \int_{\mathcal{N}^\text{o}}f\, d\left(\mathcal{J}^2_*\nu\right)= \int_{\mathcal{N}^\text{o}}f\circ \mathcal{J}\, d(\mathcal{J}_*\nu)= \int_{\mathcal{N}^\text{o}}(f\circ \mathcal{J})\rho_\nu\, d\nu\\
&\ \ \ \ \ \ \ \ \ \ \ \ \ \ \ \ \ \ \ \ \ \ \ \ \ \ \ \ \ \ \ \ \  \ \ \ = \int_{\mathcal{N}^\text{o}}f (\rho_\nu\circ  \mathcal{J})\, d(\mathcal{J}_*\nu)=
\int_{\mathcal{N}^\text{o}}f (\rho_\nu\circ  \mathcal{J}) \rho_\nu\, d\nu.
\end{align*}
for every test  function $f$. Therefore
$ \rho_\nu(\rho_\nu\circ\mathcal{J})=1$,
as claimed. For the adjoint of $\mathcal{J}$, note that
$$ \langle \mathcal{J}f, g\rangle = \int_{\mathcal{N}^{\text{o}}} (f\circ \mathcal{J}) {g}\, d\nu
=  \int_{\mathcal{N}^{\text{o}}}f ({g}\circ \mathcal{J})\, d(\mathcal{J}_*\nu)
= \int_{\mathcal{N}^{\text{o}}}f ({g}\circ \mathcal{J})\rho_\nu\, d\nu=\langle f, \mathcal{M}_{\rho_\nu}\mathcal{J}g\rangle,$$
for all test functions $f, g$. Therefore $\mathcal{J}^*=\mathcal{M}_{\rho_\nu}\mathcal{J}$.
\end{proof}
\begin{proposition}
Let  $X_0, X_1, \dots$ be the stationary Markov chain with transition kernel $B$ and $B$-invariant probability measure $\nu$. Define 
$$\left(\widetilde{B}^\ell f\right)(x):=\mathbb{E}\left[f(\mathcal{J} X_\ell)| \mathcal{J}X_{\ell+1}=x\right],$$
an operator acting on bounded continuous functions.
Then $\widetilde{B}^\ell$ does not depend on $\ell$. Dropping this superscript, 
  the Hilbert space adjoint   of $B$ and $\widetilde{B}$ are related by
$$B^*=\mathcal{J} \widetilde{B} \mathcal{J}, \ \ \widetilde{B}= \mathcal{J} B^* \mathcal{J}.$$
\end{proposition}
\begin{proof}
The first identity and the involutive property of $\mathcal{J}$ imply the second identity, which implies  that $\widetilde{B}^\ell$ does not depend on $\ell$. Let us show that   $B^*=\mathcal{J} \widetilde{B}^\ell \mathcal{J}$. The following holds  for all test functions $f, g$ on $\mathcal{N}^{\text{o}}$, and the random variables  $X, X_\ell, X_{\ell +1}$ are $\nu$-distributed:
\begin{equation*}
\begin{aligned}
\langle Bf, g\rangle &= \int_{\mathcal{N}^{\text{o}}} (Bf)(x){g}(x)\, d\nu(x)\\
&=\mathbb{E}_\nu\left[(Bf)(X){g}(X)\right]\\
&=\mathbb{E}_\nu\left[   \mathbb{E}\left[f(X_{\ell+1})|X_\ell=X\right]     {g}(X)\right]\\
&=\mathbb{E}_\nu\left[   \mathbb{E}\left[f(X_{\ell+1}) {g}(X_\ell)|X_\ell=X\right]       \right]\\
&=\mathbb{E}_\nu\left[    f(X_{\ell+1}) {g}(X_\ell)       \right]\\
&=\mathbb{E}_\nu\left[    f(X_{\ell+1}) ({g}\circ \mathcal{J})(\mathcal{J}X_\ell)       \right]\\
&=\mathbb{E}_\nu\left[   \mathbb{E}\left[ ({g}\circ \mathcal{J})(\mathcal{J}X_\ell)|\mathcal{J}X_{\ell+1} \right] f(X_{\ell+1})      \right]\\
&=\mathbb{E}_\nu\left[  \left(\widetilde{B}^{\ell}({g}\circ\mathcal{J})\right)(\mathcal{J}X_{\ell+1})   f(X_{\ell+1})      \right]\\
&=\int_{\mathcal{N}^{\text{o}}}\left(\widetilde{B}^{\ell}(g\circ\mathcal{J})\right)(\mathcal{J}x)   f(x)\, d\nu(x)  \\
&=\left\langle f, \mathcal{J}\widetilde{B}^\ell \mathcal{J} g\right\rangle.
\end{aligned}
\end{equation*}
Therefore $B^*=\mathcal{J}\widetilde{B}^\ell \mathcal{J}$ as claimed.
\end{proof}

\begin{corollary}\label{corollary_B_invariant}
If $\nu$ is $B$-invariant, then $\mathcal{J}_*\nu$ is $\widetilde{B}$-invariant.
\end{corollary}
\begin{proof}
By Proposition \ref{B*-invariant}, $\nu$ is $B^*$-invariant. Thus
$$\mathcal{J}_*\nu =\nu\mathcal{J}=\nu B^*\mathcal{J}=\nu\left(\mathcal{J}\widetilde{B} \mathcal{J}\right)\mathcal{J}= \nu\mathcal{J}\widetilde{B} 
=(\mathcal{J}_*\nu)\widetilde{B},$$
which is the invariance claimed.
\end{proof}
We may, occasionally, use the notation $\tilde{\nu}=\mathcal{J}_*\nu$.

 \subsection{Time reversal of finite chains}
 We have already used the {\em flip map} $J$. Formally,
at each regular point $q$ in $\mathcal{W}$,  we define  $J: T_qM\rightarrow T_qM$  as
 \begin{equation*}
 J(q,v)= (q,J_qv)=(q,-v)
 \end{equation*}
 and the (pointwise) {\em time reversal map}  $R_q:T_qM\times T_qM\rightarrow T_qM\times T_qM$  as
 \begin{equation*}%\label{time_reversal_first}
 R((q,u),(q,v))=R_q(u,v)=(J_qv,J_qu).
 \end{equation*}
Given a positive integer $k$,  set $[0,k]=\{j\in \mathbb{Z}: 0\leq j\leq k\}$ and define the space of finite chain segments
\begin{equation*}\mathcal{D}^{\text{o}}_{[0,k]}:=\{(x_0, x_{1}, \dots, x_{k}): (x_{j}, x_{j+1})\in \mathcal{D}^{\text{o}} \text{ for } j=0, 1, \dots, k-1\}.
\end{equation*}
Naturally, $\mathcal{D}^{\text{o}}_{[0,1]}=\mathcal{D}^{\text{o}}$.

 Using the same notation, $R$, we  define a time reversal map on finite chain segments.
\begin{definition}[Time reversal map on finite chain segments]
 We define   on the space of finite chain segments the {\em time reversal map}  $\mathcal{R}:\mathcal{D}_{[0,k]}\rightarrow \mathcal{D}_{[0,k]}$  as
\begin{equation*}
\mathcal{R}:(x_0, \dots, x_k)\mapsto (\mathcal{J}x_k, \dots, \mathcal{J}x_0).
\end{equation*}
 Note the distinction between the pointwise time-reversal map $R$ previously defined and the  finite chain reversal map $\mathcal{R}=R\circ (\mathcal{T}\times \cdots \times \mathcal{T})$.
\end{definition}

Let $X_0, X_1, \dots$ be the stationary Markov chain with transition kernel $B$ and $B$-invariant probability distribution $\nu$. 
For any positive integer $k$, set 
$$(Y_0, \dots, Y_k)=\mathcal{R}(X_0, \dots, X_k). $$
Thus $Y_j = \mathcal{J}X_{k-j}$. 
\begin{proposition}
The finite chain $(Y_0, \dots, Y_k)$ is a segment of a stationary Markov chain in state space $\mathcal{N}^{\text{\em o}}$
with transition kernel $\widetilde{B}$ and invariant measure $\tilde{\nu}:=\mathcal{J}_*\nu$, where $\nu$ is stationary for the chain generated by $B$. 
\end{proposition}
\begin{proof}
Simply note that
\begin{equation*}
\begin{aligned}
\mathbb{E}\left[f(Y_{j+1})|Y_{j}=y\right]&= \mathbb{E}\left[f(\mathcal{J}X_{k-j-1})|\mathcal{J}X_{k-j}=y\right] \\
&=\mathbb{E}\left[f(\mathcal{J}X_{\ell})|\mathcal{J}X_{\ell+1}=y\right]\\
&=\left(\widetilde{B}f\right)(y). 
\end{aligned}
\end{equation*}
This shows that $\widetilde{B}$ is the transition kernel of the chain $Y_j$. The chain has stationary measure $\tilde{\nu}$
by Corollary \ref{corollary_B_invariant}.
\end{proof}

\begin{definition}[Two-step forward  measure]
Define a probability measure $\eta$ on $\mathcal{D}^\text{\em o}$ by
$$d\eta(x,y):=d\nu(x)dB_x(y). $$
We call $\eta$ the {\em two-step forward measure}.
\end{definition}
Note that, for all
test functions $f, g$ on $\mathcal{N}^{\text{o}}$,
\begin{equation}\label{eta_fxg}
\begin{aligned}
\eta(f\times g)&=\int_{\mathcal{D}^{\text{o}}} f(x)g(y)\, d\eta(x,y)=\int_{\mathcal{N}^\text{o}} f(x) \left[\int_{\mathcal{N}^\text{o}_{\pi(\mathcal{T}(x))}} g(y) \, dB_x(y)\right]\, d\nu(x)\\
&=\langle f, Bg\rangle=\nu(fBg).
\end{aligned}
 \end{equation}

\begin{definition}[Two-step backward measure]
Define a probability measure $\tilde{\eta}$ on $\mathcal{D}^{\text{\em o}}$ by
$$\tilde{\eta}:=\mathcal{R}_*\eta. $$
We call $\tilde{\eta}$ the {\em two-step backward measure}.
\end{definition}

\begin{proposition}
The two-step backward measure satisfies
\begin{equation}\label{proposition_eta_tilde}
d\tilde{\eta}(x,y)= d(\mathcal{J}_*\nu)(x)\, d\left(\mathcal{J}B^*\mathcal{J}\right)_x(y)=d\nu(x)\, d\left(\mathcal{J}^*B^*\mathcal{J}\right)_x(y).
\end{equation}
Furthermore,
\begin{equation}\label{eta_tilde_inner_product}
\tilde{\eta}(f\times g)=\left\langle f, \mathcal{J}^*B^*\mathcal{J}g\right\rangle.
\end{equation}
for test functions $f, g$ on $\mathcal{N}^\text{\em o}$. 
\end{proposition}
\begin{proof}
Using Equation (\ref{eta_fxg}), we obtain
\begin{equation*}
\tilde{\eta}(f\times g)=\eta((f\times g)\circ \mathcal{R})=\eta((g\circ \mathcal{J})\times (f\circ\mathcal{J}))=\left\langle g\circ\mathcal{J}, B(f\circ\mathcal{J})\right\rangle=\left\langle \mathcal{J}g, B\mathcal{J}f\right\rangle.
\end{equation*}
Therefore
$
\tilde{\eta}(f\times g)=\left\langle f, \mathcal{J}^* B^* \mathcal{J}g\right\rangle$ as claimed.
Writing out the inner product explicitly,
$$\int_{\mathcal{D}^{\text{o}}} f(x)g(y)\, d\tilde{\eta}(x,y)=\tilde{\eta}(f\times g)=\int_{\mathcal{N}^\text{o}}f(x)\left[\int_{\mathcal{N}^{\text{o}}_{\pi(\mathcal{T}(x))}} g(y)\, d\left(\mathcal{J}^* B^* \mathcal{J}\right)_x(y) \right]\, d\nu(x),$$
yields
$d\tilde{\eta}(x,y)= d\nu(x) \left(\mathcal{J}^* B^* \mathcal{J}\right)_x(y).$
It remains to check the first equality in Equation (\ref{proposition_eta_tilde}). Starting from Equation (\ref{eta_tilde_inner_product}),
\begin{equation*}
\begin{aligned}
\tilde{\eta}(f\times g)&=\left\langle \mathcal{J}f,B^*\mathcal{J}g\right\rangle\\
& = \int_{\mathcal{N}^{\text{o}}} f(\mathcal{J}(x))\left(B^*\mathcal{J}g\right)(x)\, d\nu(x)\\
&=  \int_{\mathcal{N}^{\text{o}}} f(x)\left(B^*\mathcal{J}g\right)(\mathcal{J}(x))\, d(\mathcal{J}_*\nu)(x)\\
&=  \int_{\mathcal{N}^{\text{o}}} f(x)\left(\mathcal{J}B^*\mathcal{J}g\right)(x)\, d(\mathcal{J}_*\nu)(x)\\
&=  \int_{\mathcal{N}^{\text{o}}} 
f(x)\left[
\int_{
\mathcal{N}^\text{o}_{\pi(\mathcal{T}(x))}
}
g(y)\, d\left(\mathcal{J}B^*\mathcal{J}\right)_x(y)
\right]\, d(\mathcal{J}_*\nu)(x)
\end{aligned}
\end{equation*}
Consequently, $d\tilde{\eta}(x,y)=d(\mathcal{J}_*\nu)(x)\,  d\left(\mathcal{J}B^*\mathcal{J}\right)_x(y).$
\end{proof}

\begin{proposition}\label{explicit_eta_tilde}
The two-step backward measure, expressed in terms of the random scattering operator $S$, has the form
\begin{equation}\label{eta_tilde_disintegration2}
d\tilde{\eta}(x,y)=d\tilde{\nu}(y) d\widetilde{S}_y(x),
\end{equation}
in which $d\widetilde{S}_y(x):=d(\mathcal{J}_*S_{{J}(y)})(x)$ and 
 $d\tilde{\nu}(y):=d(\mathcal{J}_*\nu)(y)$.
\end{proposition}
\begin{proof}
Let $f(x,y)$ be a test function on $\mathcal{D}^{\text{o}}$. Then
\begin{equation*}
\begin{aligned}
\tilde{\eta}(f)&=\eta(f\circ\mathcal{R})\\ 
&= \int_{\mathcal{D}^\text{o}}f(\mathcal{J}(y),\mathcal{J}(x))\, dS_{\mathcal{T}(x)}(y)\, d\nu(x)\\
&=
\int_{\mathcal{D}^\text{o}}f(x,y)\, d(\mathcal{J}_*S_{\mathcal{T}(\mathcal{J}(y))})(x)\, d(\mathcal{J}_*\nu)(y).
\end{aligned}
\end{equation*}
The claimed result now follows from $\mathcal{T}\circ\mathcal{J} = J$.
\end{proof}
Note that (\ref{eta_tilde_disintegration2}) is simply the backward disintegration of the same joint measure $\tilde{\eta}$; the roles of the coordinates $x,y$ are interchanged compared to the forward disintegration (\ref{proposition_eta_tilde}).

Let us now consider   the probability distribution on finite segments of the stationary chain with transition kernel $B$ and stationary measure $\nu$,  and the probability distribution of their time reversal. We define the probability measure $P_{[0,k]}$ on $\mathcal{D}^{\text{o}}_{[0,k]}$ as
\begin{equation*}%\label{definition_Pk}
dP_{[0,k]}(x_0, \dots, x_k):=d\nu(x_0)\, dB_{x_0}(x_1)\cdots\, dB_{x_{k-1}}(x_k).
\end{equation*}
Note that if $f_0, \dots, f_k$ are bounded continuous functions on $\mathcal{N}^\text{o}$, then
\begin{equation}\label{inner_product_form_of_Pk}
\begin{aligned}
P_{[0,k]}(f_0\times \cdots \times f_k)&=\left\langle 1, f_0B\left(f_1B\left( \cdots    f_{k-2}B\left(f_{k-1}Bf_k\right)   \cdots \right)\right) \right\rangle\\
&=\nu\left(f_0B\left(f_1B\left( \cdots    f_{k-2}B\left(f_{k-1}Bf_k\right)   \cdots \right)\right) \right).
\end{aligned}
\end{equation}
We define the time reversal of $P_{[0,k]}$
as the probability measure $\widetilde{P}_{[0,k]}$ on $\mathcal{D}^{\text{o}}_{[0,k]}$ given by
\begin{equation}\label{time_reversal_measure}
\widetilde{P}_{[0,k]}=\mathcal{R}_*P_{[0,k]}.
\end{equation}
\begin{proposition}
The   measure $\widetilde{P}_{[0,k]}$ on $\mathcal{D}_{[0,k]}^{\text{\em o}}$ gives the  probability distribution  on finite chains of length $k$
of the time-reversal Markov chain, with transition kernel $\widetilde{B}$ and stationary probability $\tilde{\nu}$. Explicitly,
\begin{equation}\label{proposition_equation_P_tilde_k}
d\widetilde{P}_{[0,k]}(x_0, \dots, x_k)=d\tilde{\nu}(x_0)\, d\widetilde{B}_{x_0}(x_1)\cdots\, d\widetilde{B}_{x_{k-1}}(x_k).
\end{equation}
\end{proposition}
\begin{proof}
It is sufficient to show (\ref{time_reversal_measure}) when integrating a product function $f_0(x_0)\cdots f_k(x_k)$ with respect to $\widetilde{P}_{[0,k]}$.
\begin{equation*}
\begin{aligned}
\widetilde{P}_{[0,k]}(f_0\times \cdots \times f_k)&=P_{[0,k]}(\mathcal{J}f_k\times\cdots \times \mathcal{J}f_0)\\
&=\left\langle 1, (\mathcal{J}f_k)B\left((\mathcal{J}f_{k-1})B\left( \cdots    (\mathcal{J}f_{2})B\left((\mathcal{J}f_{1})B\mathcal{J}f_0\right)   \cdots \right)\right) \right\rangle,
\end{aligned}
\end{equation*}
where we have used Equation (\ref{inner_product_form_of_Pk}). Taking adjoints,
\begin{equation*}
\widetilde{P}_{[0,k]}(f_0\times \cdots \times f_k)=\left\langle 1, (\mathcal{J}f_0)B^*\left((\mathcal{J}f_{1})B^*\left( \cdots    (\mathcal{J}f_{k-2})B^*\left((\mathcal{J}f_{k-1})B^*\mathcal{J}f_k\right)   \cdots \right)\right) \right\rangle,
\end{equation*}
and using the identity 
\begin{equation*}
(\mathcal{J}f)B^*(\mathcal{J}g)=\mathcal{J}\left(f \mathcal{J}B^*\mathcal{J}g\right)=\mathcal{J}\left(f\widetilde{B}g\right),
\end{equation*}
we arrive at
\begin{equation*}
\widetilde{P}_{[0,k]}(f_0\times \cdots \times f_k)=\left\langle 1,\mathcal{J}\left( f_0\widetilde{B}\left(f_1\widetilde{B}\left( \cdots    f_{k-2}\widetilde{B}\left(f_{k-1}\widetilde{B}f_k\right)   \cdots \right)\right) \right)\right\rangle.
\end{equation*}
Note that $\langle 1, \mathcal{J}g\rangle = \tilde{\nu}(g)$. This expression implies Equation (\ref{proposition_equation_P_tilde_k}).
\end{proof}

\begin{proposition}
The following identities involving Radon-Nikodym derivatives hold:
\begin{equation*}
\frac{d\tilde{\nu}}{d\nu}(x) \frac{d\widetilde{B}_x}{dB_x}(y)= \frac{d\tilde{\eta}}{d\eta}(x,y)
\end{equation*}
and 
\begin{equation*}
\frac{d\widetilde{P}_{[0,k]}}{dP_{[0,k]}}(x_0, \dots, x_k) = \frac{d\tilde{\nu}}{d\nu}(x_0) \frac{d\widetilde{B}_{x_0}}{dB_{x_0}}(x_1)\cdots
\frac{d\widetilde{B}_{x_{k-1}}}{dB_{x_{k-1}}}(x_k).
\end{equation*}
From these it follows that
\begin{equation}\label{radon_nikodym_P_tilde_P}
\frac{d\widetilde{P}_{[0,k]}}{dP_{[0,k]}}(x_0, \dots, x_k) = \left(\prod_{j=1}^{k-1} \rho_{\nu}(x_j)\right)^{-1}  \frac{d\tilde{\eta}}{d\eta}(x_0,x_1)\cdots
 \frac{d\tilde{\eta}}{d\eta}(x_{k-1},x_k).
\end{equation}
\end{proposition}
\begin{proof}
The first two identities are obtained  straightforwardly by integrating product test functions against the measures on both sides.
The third identity is an immediate consequence of the first two and the definition of $\rho_\nu$. 
\end{proof}
\begin{proposition}\label{proposition_collapsing}
Let $f$ be a test function   on $\mathcal{D}^{\text{\em o}}_{[0,1]}$.
Then
\begin{equation*}
\int_{\mathcal{D}_{[0,k]}^{\text{\em o}}} f(x_i, x_{i+1})\, dP_{[0,k]}(x_0, \dots, x_k) = \int_{\mathcal{D}^{\text{\em o}}_{[0,1]}}f(x,y)\, d\eta(x,y) = \eta(f).
\end{equation*}
If $f$ is a test function on $\mathcal{N}^{\text{\em o}}$, then
\begin{equation*}
\int_{\mathcal{D}_{[0,k]}^{\text{\em o}}} f(x_i)\, dP_{[0,k]}(x_0, \dots, x_k) = \int_{\mathcal{N}^{\text{\em o}}}f(x)\, d\nu(x) = \nu(f).
\end{equation*}
\end{proposition}
\begin{proof}
We establish the first identity. The second is proved in a similar manner.  As in other such identities, it is sufficient to 
consider product functions of the form 
\begin{equation*}f(x_i,x_{i+1})={1}(x_0)\cdots {1}(x_{i-1}) f_i(x_i)f_{i+1}(x_{i+1}) {1}(x_{i+2})\cdots {1}(x_k),\end{equation*}
where $1(x)$ indicates the constant function in the variable $x$.
Using Equation (\ref{inner_product_form_of_Pk}) and noting that $B{1}={1}$ and $\nu B=\nu$, we obtain
\begin{equation*}
\int_{\mathcal{D}_{[0,k]}^{\text{\em o}}} f(x_i, x_{i+1})\, dP_{[0,k]}(x_0, \dots, x_k) = \nu(f_iBf_{i+1})=\eta(f_i\times f_{i+1}).
\end{equation*}
At the last step  we used Equation (\ref{eta_fxg}).
\end{proof}

  \subsection{Entropy production in the stationary process}
The following definitions are taken from \cite{qian}.
\begin{definition}[Relative entropy] Suppose $P_1$ and $P_2$ are two probability measures on a measurable space $(\mathcal{D},\mathcal{F})$. The {\em relative entropy} of $P_1$ with respect to $P_2$ is defined as
\begin{equation*}
H(P_1,P_2):=\begin{cases}
\int_{\mathcal{D}}\log\left(\frac{dP_1}{dP_2}\right)\, dP_1 & \text{if } P_1 \ll P_2 \text{ and } \log\left(\frac{dP_1}{dP_2}\right)\in L^1(\mathcal{D},P_1)\\
+\infty & \text{otherwise}.
\end{cases}
\end{equation*}
\end{definition}

\begin{definition}[Entropy production rate] The {\em entropy production rate} of the stationary Markov chain $X_0, X_1, \dots$ defined by the transition kernel  $B$ and stationary probability measure $\nu$ is defined by
\begin{equation*}
e_p:=\lim_{n\rightarrow \infty} \frac1n H\left(P_{[0,n]},\widetilde{P}_{[0,n]}\right),
\end{equation*}
where $H\left(P_{[0,n]},\widetilde{P}_{[0,n]}\right)$ is the relative entropy of $P_{[0,n]}$ with respect to $\widetilde{P}_{[0,n]}$ restricted to
the $\sigma$-algebra generated by $X_0, \dots, X_n$. 
\end{definition}

Let us return to the setting of the previous section. We assume here absolute continuity between measures whose Radon-Nikodym derivatives 
are being used. 
\begin{theorem}[Entropy production rate formula]\label{theorem_main_ep}
The entropy production rate for the stationary random dynamical system defined by the transition operator $B$ and stationary probability measure $\nu$, under the assumption that the pairs $\nu, \tilde{\nu}$ and $\eta, \tilde{\eta}$ define equivalent measure classes, has the form
\begin{equation}\label{theorem_ep}
e_p=\int_{\mathcal{N}^\text{\em o}} 
\log\rho_\nu\, d\nu - 
\int_{\mathcal{D}^{\text{\em o}}} \log\left(\frac{d\tilde{\eta}}{d\eta}\right)\, 
d\eta=\int_{\mathcal{D}^\text{o}}\log\left(\rho_\nu(x)\frac{d\eta}{d\tilde{\eta}}(x,y)\right)\, d\eta(x,y). 
\end{equation}  
\end{theorem}
\begin{proof}
A direct consequence of Equation (\ref{radon_nikodym_P_tilde_P}) and the identities in Proposition \ref{proposition_collapsing} is
\begin{equation*}
H\left(P_{[0,n]},\widetilde{P}_{[0,n]}\right)=(n-1)\int_{\mathcal{N}^\text{o}} \log\rho_\nu\, d\nu -n\int_{\mathcal{D}^\text{o}} \log\left(\frac{d\tilde{\eta}}{d\eta}\right)\, d\eta.
\end{equation*}
Dividing by $n$ and passing to the limit $n\rightarrow \infty$ yields Equation (\ref{theorem_ep}).
\end{proof}

\begin{corollary} The entropy production rate $e_p$ for the stationary system defined by $\nu$ and $B$ as in Theorem \ref{theorem_main_ep}, Equation (\ref{theorem_ep}) satisfies
\begin{equation*}
0\leq e_p = H(\eta\| \tilde{\eta}) -H(\nu\|\tilde{\nu}),
\end{equation*}
where $H(\mu_1\| \mu_2)$ is the Kullback-Leibler divergence (relative entropy) of the probability measures $\mu_1, \mu_2$.
\end{corollary} 
\begin{proof}
This is a simple matter of definitions. That $e_p$ is non-negative is a standard property of the relative entropy known as Gibbs's inequality. 
\end{proof}
Note that the non-negative quantity $e_p$ is the difference of two non-negative quantities. 

The above expressions for $e_p$ may be written in a more symmetric form that  makes the non-negativity of each term more explicit:
\begin{equation*}
e_p=\frac12 \int_{\mathcal{D}^{\text{o}}}\log\left(\frac{d\tilde{\eta}}{d\eta}\right)\left(d\tilde{\eta} -d{\eta}\right)-\frac12\int_{\mathcal{N}^{\text{o}}}
 \log\left(\frac{d\tilde{\nu}}{d\nu}\right)\left(d\tilde{\nu}-d\nu\right).
\end{equation*}
This follows from the previously seen identity $\rho_\nu\circ\mathcal{J}=\rho_\nu^{-1}$ and the similarly easy to establish identity
$\frac{d\tilde{\eta}}{d\eta}\circ \mathcal{R}=\frac{d{\eta}}{d\tilde{\eta}}.$
 
\section{Introducing temperature}\label{sec:introducing_temperature}
Our discussion so far  applies to general models of random impact/transmission. It is  now time to introduce wall temperature. This is done, as in the literature of the Boltzmann equation, through the notion of {\em reciprocity} with respect to Maxwellian probability distributions.
 \subsection{Surface Maxwellians and reciprocity}
 \begin{definition}[Surface Maxwellian at temperature $T$]
 The {\em Maxwellian} (or Maxwell-Boltzmann distribution) at $q\in \mathcal{W}$ (on either side of a wall) and temperature $T(q)$ is
 the probability measure $\mu_q\in \mathcal{P}(W_q^\pm)$ having density relative to Lebesgue measure given by
 \begin{equation}\label{maxwellian_definition}
 \rho_q(v)=2\pi\left(\frac{\beta(q)m}{2\pi}\right)^{\frac{n+1}{2}}|\langle v, \mathbf{n}_q\rangle| \exp\left\{-\beta(q)\frac{m|v|_q^2}{2}\right\},
 \end{equation}
 where $m$ is  particle mass, $\beta(q)=1/\kappa T(q)$, $n$ is the  dimension of $M$,  and $\kappa$ is known as the Boltzmann constant. 
Denoting by $V_q$  the  volume measure on $T_qM$, we can write $d\mu_q(v) =\rho_q(v)\, dV_q(v)$.
When necessary to distinguish the Maxwellian on different sides of a wall, we write $\mu_q^+, \mu_q^-$ for the probability measures on
$W_q^+$ and $W_q^-$, respectively. 
 \end{definition}
 
The term $|\langle v,\mathbf{n}_q\rangle|$ in Equation (\ref{maxwellian_definition}) distinguishes  surface from  bulk Maxwellians. Because of this term, 
 the  probability distribution of directions  has a cosine density,  often referred to as the {\em Knudsen cosine law}. 
The subscript $q$ may  be dropped from $\mu_q$ if there is no risk of confusion.  The same applies to the density $\rho_q(v)$.

 Recall the {\em time reversal map}  $R_q:T_qM\times T_qM\rightarrow T_qM\times T_qM$ at a regular point $q$, which is defined as
 \begin{equation*}%\label{time_reversal}
 R((q,u),(q,v))=R_q(u,v)=(J_qv,J_qu).
 \end{equation*}
 Given a general random scattering operator $S$, we define at  $q\in \mathcal{W}$ the probability measure 
 $\zeta_q\in \mathcal{P}(T_qM\times T_qM)$ on pairs $(v^\text{i},v^\text{o})$ of incoming-outgoing velocities as 
 \begin{equation}\label{zeta_in_out}
 d\zeta_q(v^\text{i},v^\text{o}):=d\mu_q(v^\text{i}) dS_{\left(q,v^\text{i}\right)}(v^\text{o}) :=d\mu_q(v^{\text{i}}) S_q(dv^{\text{o}}|v^\text{i}).
 \end{equation}
 Here, and in the below Definition \ref{definition_reciprocity},  we allow $\mu_q$ to be any convex combination
  \begin{equation*}\mu_q=a\mu_q^-+(1-a)\mu_q^+,\end{equation*}
  where  $0\leq a\leq 1$.

Wall temperature is introduced in our compartmented system  by imposing {\em reciprocity}, as defined next. 
\begin{definition}[Reciprocity]\label{definition_reciprocity}
The general random  scattering operator $S$ has the property of {\em reciprocity} if at each regular  boundary point $q$ the probability measure $\zeta_q$ defined in 
(\ref{zeta_in_out}) is invariant under the time-reversal map $R_q$. The operator  $S$ satisfying reciprocity with respect to a surface Maxwellian having temperature parameter $T$ will be   called a 
{\em random  scattering operator at temperature} $T$.
\end{definition}
  Since $J_*\mu_q^\pm =\mu_q^\mp$, the reciprocity condition for $S$ with respect to $\mu=\frac12 \mu_q^-+\frac12\mu_q^+$
  amounts to
  \begin{equation*}%\label{equation_reciprocity_definition}
  d\mu_q(v^\text{i})\, dS_{(q,v^\text{i})}(v^\text{o}) =d\mu_q(Jv^\text{o})\, dS_{(q,Jv^\text{o})}(Jv^\text{i}).
  \end{equation*}
 
By this definition, individual wall-collision events satisfy a detailed balance condition relative to the surface Maxwellian distribution of velocities (with temperature parameter given by the specified wall temperature at $q$).  
In an isolated compartment enclosed by walls with constant temperature $T$, the particle motion  defines a Markov process which, if irreducible, will have for its unique stationary velocity distribution the  Maxwellian distribution
with the same temperature $T$, and no entropy or work will be produced in the stationary regime. See  \cite{cook} for a large class of  random-mechanical  reflection operators $S$ satisfying this property. Simpler examples will be given shortly. 
Reciprocity is  typically assumed for boundary operators defining boundary conditions for the Boltzmann equation. See, for example, \cite{cercignani}.
 
  \begin{figure}[ht]
\begin{center}
\includegraphics[width=2 in]{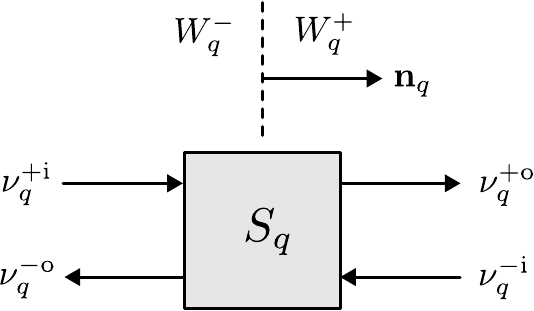}\ \ 
\caption{{\small  Particle scattering at a wall. The superscript $\pm$ is defined in relation to the wall orientation $\mathbf{n}_q$ while $i$ and $o$ stand for `incoming' (pre-collision) and `outgoing' (post-collision). The $\nu_q^{\pm\text{i}}$ and $\nu_q^{\pm\text{o}}$ are sub-probability measures.}}
\label{S_block}
\end{center}
\end{figure}

 It is convenient to express the random scattering operator more explicitly in terms of reflection and transmission probabilities. 
 For this, let us write an incoming (or pre-collision) probability measure as a pair $\nu^{\text{i}}_q=(\nu_q^{-\text{i}}, \nu_q^{+\text{i}})$, where $\nu_q^{\epsilon\text{i}}$
 is supported in $W_q^\epsilon$, $\epsilon\in \{-,+\}$. Similarly, we write the outgoing (or post-collision) probability measure as  $\nu^{\text{o}}_q=(\nu_q^{-\text{o}}, \nu_q^{+\text{o}})$, where $\nu^{\epsilon \text{o}}$ is also supported in $W_q^{\epsilon}$.
 The scattering map on probability measures (see (\ref{nu_S})) then assumes the block form
 \begin{equation*}
(\nu_q^{-\text{o}}, \nu_q^{+\text{o}})=(\nu_q^{-\text{i}}, \nu_q^{+\text{i}})\left(\begin{array}{cc}S_q^{--} & S_q^{-+} \\S_q^{+-} & S_q^{++}\end{array}\right).
 \end{equation*}
  See Figure \ref{S_block}.
 More explicitly, the expression
 \begin{equation*}
 \nu_q^{\epsilon\text{o}} = \nu_q^{-\text{i}}S_q^{-\epsilon} + \nu_q^{+\text{i}}S_q^{+\epsilon}
 \end{equation*}
 should be understood as the sub-probability measure 
 \begin{equation*}
  d\nu_q^{\epsilon\text{o}}(v) = \int_{W_q^-}d\nu_q^{-\text{i}}(u)S_q^{-\epsilon}(dv|u) + \int_{W_q^+}d\nu_q^{+\text{i}}(u)S_q^{+\epsilon}(dv|u).
 \end{equation*}
 Thus, if $U\subset W_q^{\epsilon}$ is a measurable set, 
 \begin{equation*}
 \nu_q^{\epsilon\text{o}}(U) =\int_{W_q^-}d\nu_q^{-\text{i}}(u)S_q^{-\epsilon}(U|u) + \int_{W_q^+}d\nu_q^{+\text{i}}(u)S_q^{+\epsilon}(U|u).
 \end{equation*}
 
 A simple but useful observation is the following conservation identity:
 \begin{proposition}\label{proposition_balance}
 For any random scattering operator $S$ and for each $q\in \mathcal{W}$, we have  
 \begin{equation*}
 \nu_q^{-\text{\em o}}(W_q^-)+\nu_q^{+\text{\em o}}(W_q^+)= \nu_q^{-\text{\em i}}(W_q^-)+\nu_q^{+\text{\em i}}(W_q^+).
 \end{equation*}
 \end{proposition}
 \begin{proof}
 This identity follows immediately from  
 $S^{\epsilon -}_q(W^-_q|v)+S^{\epsilon +}_q(W^+_q|v)=1.$
 \end{proof}
 
 This block form of $S$ differs in an inconsequential way from how scattering matrices are conventionally written in other subjects, but it is natural in our setting.
 Here  the block $S_q^{--}$ represents transmission of an incoming velocity vector  in $W_q^-$  from the side of the wall towards  which $\mathbf{n}_q$  is pointing   to the other side, thus with outgoing vector still in $W_q^-$; $S_q^{-+}$ represents reflection  of a velocity vector  from $W_q^-$ into $W_q^+$, and so on.

The probability measure $\zeta_q$ defined  by Equation (\ref{zeta_in_out}) breaks up into sub-probability measures describing transmission and reflection:
\begin{equation}\label{definition_zeta_epsilon1_epsilon2}
\zeta^{\epsilon_1\epsilon_2}(A):=\zeta\left(A\cap \left(W^{\epsilon_1}\times W^{\epsilon_2}\right)\right),
\end{equation}
for any measurable subset $A$ of $\mathcal{D}_q^{\text{o}}$ and $\epsilon_1, \epsilon_2\in \{-,+\}$.  
\begin{proposition}
Let $\mu_q =\mu_q^{\epsilon_1}$ (so that either $a=0$ or $a=1$). Then we have
\begin{equation}\label{zeta_expression}
d\zeta^{\epsilon_1\epsilon_2}_q(v^\text{i},v^\text{o})=d\mu^{\epsilon_1}_q(v^\text{i}) S^{\epsilon_1\epsilon_2}_{q}(dv^\text{o}|v^\text{i}),
\end{equation}
and the reciprocity condition, $R_*\zeta=\zeta$,   is equivalent to
\begin{equation}\label{zeta_epsilon1_epsilon2}
R_*\zeta^{-+}=\zeta^{-+}, \ \  R_*\zeta^{+-}=\zeta^{+-}, \ \  R_*\zeta^{--}=\zeta^{++}, \ \ R_*\zeta^{++}=\zeta^{--}, 
\end{equation}
in which the subindex $q$ has been dropped for simplicity. 
\end{proposition}
\begin{proof}
Expression (\ref{zeta_expression}) follows   directly from definitions. For the second part, we only note that 
\begin{equation*}
\zeta\left(\mathbbm{1}_{W^{\epsilon_1}\times W^{\epsilon_2}}(f\circ R)\right)=
\zeta\left(\left(\mathbbm{1}_{W^{-\epsilon_2}\times W^{-\epsilon_1}}f\right)\circ R\right)=
\zeta\left(\mathbbm{1}_{W^{-\epsilon_2}\times W^{-\epsilon_1}}f\right)
\end{equation*}
where $f$ is any test function on $\mathcal{D}^{\text{o}}$ and   the last equality holds under the
assumption of reciprocity.
The equivalence of the reciprocity condition and (\ref{zeta_epsilon1_epsilon2}) is obtained from this remark and
 (\ref{definition_zeta_epsilon1_epsilon2}).
\end{proof}

In what follows, we define the probability  measure $\mu$ on the bundle $\mathcal{N}^{\text{o}}$ over $\mathcal{W}$ such that for each $x=(q,v)$,
\begin{equation}\label{max_minus_o}
d\mu(x)=d\mu_q(v)\, d\bar{A}(q).
\end{equation}
Here $\bar{A}$ is the normalized Riemannian volume measure on $\mathcal{W}$ (the area measure in dimension $3$).
The velocity  vector $v$ may lie on either side of a wall $W$. On the $\pm$-side, we set $d\mu_q(v)=d\mu^\pm_q(v)$, 
 the surface Maxwellian with temperature parameter $T(q)$.  We may add superscripts $\mu^{\text{i}}$ and $\mu^{\text{o}}$
 when referring to incoming or outgoing velocities. Note that, on a boundary wall $W$, the measure $\mu^{\text{i}}_q$ has support in $W_q^{-}$
 and  $\mu^{\text{o}}$ has support in $W_q^+$. On separating walls, incoming or outgoing vectors may lie on either side of $W$. 

\subsection{Entropy production, energy flux and temperature}
The next proposition is the key ingredient connecting the general expression for entropy production in the stationary random flight system
defined by a general random scattering operator $S$ and the more familiar thermodynamic form, given  in terms of
energy flux over temperature, when the walls are assigned a temperature function through the assumption of reciprocity
with respect to Maxwellians. 
\begin{proposition}\label{proposition_delta_delta_tilde}
With the definitions given above, and under the assumption of reciprocity for the random scattering operator,  we have
\begin{equation*}
\frac{d\eta}{d\tilde{\eta}}(x,y)=\frac{d\nu}{d\left(\mathcal{J}_*\mu^\text{\em o}\right)}(x)\left[\frac{d\nu}{d\left(\mathcal{J}_*\mu^\text{\em o}\right)}(\mathcal{J}y)\right]^{-1}
\end{equation*}
for $(x,y)\in \mathcal{D}^{\text{\em o}}$. 
\end{proposition}
\begin{proof}
This is, essentially, Proposition 5 from \cite{ChumFer21}, although here we expand on a few points which  that paper should have included for clarity. For simplicity, we drop the superscript $\text{o}$ from the surface Maxwellian $\mu^\text{o}$ (Equation (\ref{max_minus_o})) and write
$\tilde{\mu}=\mathcal{J}_*\mu$.  Let $\Lambda:=\frac{d\nu}{d\tilde{\mu}}$.  We write the measure with element $d\nu=\Lambda\, d\tilde{\mu}$
as $\nu=\tilde{\mu}\circ \mathcal{M}_\Lambda$, where $\mathcal{M}_\Lambda$ is multiplication by $\Lambda$ and $\tilde{\mu}$ is the linear functional defined by integration. 
Because $\mathcal{J}_*\tilde{\mu}=\mathcal{J}_*\mathcal{J}_*\mu=\mu$,
the time-reversed stationary measure is
\begin{equation}%\label{equation_dtilde_nu_dmu}
\tilde{\nu}= \mathcal{J}_*\nu=\nu\circ \mathcal{J}=\tilde{\mu}\circ \mathcal{M}_\Lambda\circ \mathcal{J}= \tilde{\mu}\circ\mathcal{J}\circ \mathcal{M}_{\Lambda\circ\mathcal{J}}= (\mathcal{J}_*\tilde{\mu})\circ\mathcal{M}_{\Lambda\circ\mathcal{J}}=\mu\circ \mathcal{M}_{\Lambda\circ \mathcal{J}}.
\end{equation}
Thus
\begin{equation}\label{equation_dtilde_nu_dmu}
d\tilde{\nu} = \Lambda\circ \mathcal{J}\, d\mu.
\end{equation}
 Recall our typical notation $x=(q,v)$ for elements of $\mathcal{N}^\text{o}$. For $(x,y)\in \mathcal{D}^{\text{o}}$, so that $x, y\in \mathcal{N}^\text{o}$ with $\pi(y)=\pi(\mathcal{T}(x))$,
 \begin{equation*}
 B_x(dy)=S_{\mathcal{T}(x)}(dy).
 \end{equation*}
 For any test function $f$ on $\mathcal{D}^{\text{o}}$, 
 \begin{equation*}
 \begin{aligned}
 \eta(f)&=\int_{\mathcal{D}^\text{o}} f(x,y) \Lambda(x)\, d\tilde{\mu}(x)\, dS_{\mathcal{T}(x)}(y)\\
 &=\int_{\hat{\mathcal{D}}^\text{o}} f(\mathcal{J}(z),y) \Lambda(\mathcal{J}(z))\, d{\mu}(z)\, dS_{\mathcal{T}(\mathcal{J}(z))}(y),
 \end{aligned}
 \end{equation*}
 where $\hat{\mathcal{D}}^\text{o}$ is the image of $\mathcal{D}^\text{o}$ under $\mathcal{J}\times I$. Note that $\mathcal{T}\circ\mathcal{J}=J$.
 Hence
 \begin{equation*}
 \eta(f)=\int_{\hat{\mathcal{D}}^\text{o}} f(\mathcal{J}(z),y)\Lambda(\mathcal{J}(z))\, d\mu(z)\, dS_{Jz}(y).
 \end{equation*}
 Similarly,
 \begin{equation*}
 \begin{aligned}
 \tilde{\eta}(f)&=\int_{\mathcal{D}^\text{o}} f(x,y)\, d\tilde{\nu}(y)\, d\tilde{S}_y(x)\\
 &=\int_{\mathcal{D}^\text{o}} f(x,y)\, \Lambda(\mathcal{J}(y))\, d{\mu}(y)\, d(\mathcal{J}_*S)_{Jy}(x)\\
 &=\int_{\hat{\mathcal{D}}^\text{o}} f(\mathcal{J}(z),y)\Lambda(\mathcal{J}(y))\, d\mu(y)\, dS_{Jy}(z).
 \end{aligned}
 \end{equation*}
 Let us define the measure $\Gamma$ on pairs $(z,y)\in \hat{\mathcal{D}}^\text{o}$ by 
 \begin{equation*}
 d\Gamma(z,y)=d\mu^\text{o}(z)\, dS_{Jz}(y). 
 \end{equation*}
 (Here it is convenient, for clarity, to reintroduce the superscript $\text{o}$.) Note that the reciprocity condition,
 expressed in terms of $\mu^{\text{i}}$ and $\mu^{\text{o}}$,
implies that $\Gamma$ is symmetric: on test functions $f$,
 \begin{equation*}
\int_{\hat{\mathcal{D}}^{\text{o}}} f(z,y)\, d\Gamma(z,y) = \int_{\hat{\mathcal{D}}^{\text{o}}} f(y,z)\, d\Gamma(z,y).
 \end{equation*}
 This follows from $\mu^\text{o} =J_*\mu^{\text{i}}$ and $R_*\zeta=\zeta$ (as in Definition \ref{definition_reciprocity}). 
 More explicitly, note that $\Gamma=(J\times I)_*\zeta$ and, introducing  $\Pi(x,y)=(y,x)$ ($x, y$ at the same base-point), we have
 \begin{equation*}
 \Pi\circ (J\times I)=(J\times I)\circ R,
 \end{equation*}
from which we obtain
 \begin{equation*}
 \Pi_*\Gamma = \Pi_*(J\times I)_*\zeta = (J\times I)_*R_*\zeta = (J\times I)_*\zeta=\Gamma.
 \end{equation*}
 
 Let $F=\mathcal{J}\times I:\mathcal{D}^\text{o}\rightarrow \hat{\mathcal{D}}^\text{o}$.
 We can now compare $\eta$ and $\tilde{\eta}$ using the symmetry of $\Gamma$:
 \begin{equation*}
 \begin{aligned}
 \tilde{\eta}(f)&=\int_{\hat{\mathcal{D}}^\text{o}} f(\mathcal{J}(z),y)\Lambda(\mathcal{J}(y))\, d\Gamma(z,y)=\int_{\mathcal{D}^\text{o}}f(x,y)\Lambda(\mathcal{J}(y))\, d(F_*\Gamma)(x,y)\\
 {\eta}(f)&=\int_{\hat{\mathcal{D}}^\text{o}} f(\mathcal{J}(z),y)\Lambda(\mathcal{J}(z))\, d\Gamma(z,y)=
 \int_{\mathcal{D}^\text{o}}f(x,y)\Lambda(x)\, d(F_*\Gamma)(x,y).
 \end{aligned}
 \end{equation*}
 These identities express $\eta$ and $\tilde{\eta}$ in terms of the same base measure, $F_*\Gamma$. It follows that their Radon-Nikodym derivative is the ratio of densities:
$\frac{d\eta}{d\tilde{\eta}}=\Lambda(x)/\Lambda(\mathcal{J}(y))$, as claimed.
  \end{proof}

Let us now apply the general expression for entropy production (Equation \eqref{theorem_ep}) to the present setting. The following corollary 
should be compared with Proposition 6 of \cite{ChumFer21}, where a factor $1/2$ appears by mistake.
\begin{corollary}
Define  $\Lambda:=\frac{d\nu}{d\tilde{\mu}}$ as in the proof of Proposition \ref{proposition_delta_delta_tilde}.
Then 
\begin{equation*}
e_p=\nu\left[\log\left(\rho_\nu\frac{\Lambda}{\Lambda\circ \mathcal{J}}\right)\right]=
\nu\left[\log\left(\frac{d\mu}{d\tilde{\mu}}\right)\right],
\end{equation*}
where, we recall, $\rho_\nu=\frac{d\tilde{\nu}}{d\nu}$.
\end{corollary}
\begin{proof}
We apply Theorem \ref{theorem_main_ep}, noting that
\begin{equation*}
\begin{aligned}
\int_{\mathcal{D}^\text{o}} \log\left(\frac{d\eta}{d\tilde{\eta}}\right)\, d\eta&=\int_{\mathcal{D}^\text{o}} \log\left(\frac{\Lambda(x)}{\Lambda(\mathcal{J}(y))}\right)\, dB_x(y)\, d\nu(x)\\
&=\nu\left(\log\left(\Lambda\right) B  1\right) - \nu\left(B\log\left(\Lambda\circ\mathcal{J}\right)\right)\\
&=\nu\left(\log\Lambda\right) - \nu\left(\log\Lambda\circ\mathcal{J}\right),
\end{aligned}
\end{equation*}
in which we used  $B 1=1$ and $\nu B=\nu$. This gives the first equality.
For the second, 
\begin{equation*}
\rho_\nu\frac{\Lambda}{\Lambda\circ \mathcal{J}}=\frac{d\tilde{\nu}}{d\nu}\frac{\frac{d\nu}{d\tilde{\mu}}}{\frac{d\tilde{\nu}}{d\mu}}=\frac{d\mu}{d\tilde{\mu}},
\end{equation*}
in which  Equation (\ref{equation_dtilde_nu_dmu}) has been used.
\end{proof}
 
The above corollary is what is needed to obtain the classical thermodynamical interpretation of the  entropy production rate.
Before proceeding, a few remarks concerning  surface Maxwellians are needed. Recall that the measure $\mu_q$ has the form
\begin{equation*}
d\mu_q(v)= C(q) \exp\left\{-\beta(q)\frac{m|v|^2_q}{2}\right\} |\langle v,\mathbf{n}_q \rangle|\, dm(q,v),
\end{equation*}
where the reference measure  $dm(q,v)=dV_q(v)\, dA(q)$ is the Riemannian measure on the half-space bundle over walls; $dA(q)$ is the hypersurface volume of walls and $dV_q(v)$ is the Lebesgue measure on the half-space $W_q^\pm$. The measure
\begin{equation*}
d\mu_{\text{bill}}(q,v)=|\langle v,\mathbf{n}_q\rangle|_q\, dV_q(v)\, dA(q)
\end{equation*}
is the invariant billiard measure (surface Liouville measure). It is invariant under the return-to-wall map $\mathcal{T}$ and under the
inversion map $\mathcal{J}$. Since  the reciprocity condition is defined pointwise (for each $q\in \mathcal{W}$), we could have used
 more generally, in place of $\mu$ in our previous calculations, the Gibbs measure
\begin{equation*}
 d\mu_{\text{Gibbs}}(q,v)=C(q)\exp\left\{-\beta(q) E(q,v)\right\}\, d\mu_{\text{bill}}(q,v),
 \end{equation*}
 where 
 \begin{equation*}
 E(q,v):=\frac{m|v|^2_q}{2}+\Phi(q)
 \end{equation*}
 is the total energy of the system, which now includes a potential function $\Phi(q)$. 
Defining
 \begin{equation*}
 \tilde{\mu}_{\text{Gibbs}}:=\mathcal{J}_*\mu_{\text{Gibbs}},
 \end{equation*}
 we have
 \begin{equation*}
 \frac{d\tilde{\mu}_{\text{Gibbs}}}{d\mu_\text{bill}}= \frac{d(\mathcal{J}_*{\mu}_{\text{Gibbs}})}{d(\mathcal{J}_*\mu_\text{bill})}=
  \frac{d{\mu}_{\text{Gibbs}}}{d\mu_\text{bill}}\circ \mathcal{J}.
 \end{equation*}
 Furthermore, 
 \begin{equation*}
 \log\frac{d\mu_{\text{Gibbs}}}{d\mu_{\text{bill}}}(x)= -\beta(q)E(x), 
 \end{equation*}
 for $x=(q,v)$.
Also note that $E$ is invariant under the sign-change map $J$. Thus, writing $\nu^\text{i}:=\mathcal{T}_*\nu$ and
$\nu^\text{o}:=\nu$,
\begin{equation*}
\begin{aligned}
e_p&=\nu(-\beta E +(\beta E)\circ \mathcal{T})=\nu^\text{i}(\beta E)- \nu^\text{o}(\beta E)\\
& = \int_{\mathcal{N}^\text{o}} \beta(q) E(q,v)\left[d\nu^{\text{i}}(q,v)-d\nu^{\text{o}}(q,v)\right].
\end{aligned}
\end{equation*}
 We register this result in the following theorem:
 \begin{theorem}\label{second_law}
 With the above definitions, the entropy production rate for the stationary  random flight with energy function  $E$ and stationary measure $\nu$
 is
 \begin{equation}\label{main_entropy_equation}
 e_p = \int_{\mathcal{N}^\text{\em o}} \frac{E(q,v)}{\kappa T(q)}\left[d\nu^{\text{\em i}}(q,v)-d\nu^{\text{\em o}}(q,v)\right]\geq 0. 
 \end{equation}
 Interpreting $\nu^{\text{\em i}}$ and $\nu^{\text{\em o}}$ as probability fluxes, $e_p$ can be regarded as 
 the expected value of the rate of energy being transferred to the walls divided by the temperature at the point. 
 \end{theorem} 
 
 An alternative to the expression (\ref{main_entropy_equation}) for $e_p$ will be useful. 
 
 \begin{corollary}\label{corollary_of_main_theorem}
On the bundle $\mathcal{D}^{\text{\em coll}}=\{(x,y)\in \mathcal{N}^\text{\em i}\times \mathcal{N}^\text{\em o}: \pi(x)=\pi(y)\}$ over $\mathcal{W}$, define
the measure $$ d\zeta_{\text{stat}}(x,y)=d\nu^{\text{\em i}}(x) S_x(dy).$$
This is the joint distribution  of incoming-outgoing pairs of states at wall collisions in the stationary process. Define the {\em entropy change at a collision} as the function $\sigma:\mathcal{D}^{\text{\em coll}}\rightarrow \mathbb{R}$ such that
$$ \sigma(x,y)=\frac{E(x)-E(y)}{\kappa T(q)},$$
where $q=\pi(x)=\pi(y)\in \mathcal{W}$. The numerator of $\sigma(x,y)$ is the energy (positive or negative) transferred to the wall
for the incoming-outgoing pair $(x,y)$.  Then
\begin{equation*}
e_p = \int_{\mathcal{D}^\text{\em coll}} \sigma(x,y)\, d\zeta_\text{\em stat}(x,y).
\end{equation*}
Thus the entropy production rate is the expected value of the entropy change function at wall collisions in the stationary process.
 \end{corollary}
\begin{proof}
Stationarity, $\nu=\nu B$, and the definition $B=\mathcal{T}\circ S$ imply that for any test function $f$, 
\begin{equation*}
\int_{\mathcal{N}^\text{o}} f(x)\, d\nu^\text{o}(x) = \int_{\mathcal{N}^{\text{o}}}(Bf)(x)\, d\nu^{\text{o}}(x)= \int_{\mathcal{N}^{\text{o}}}\left[\int_{T_{\pi(\mathcal{T}(x))}M} f(y) S_{\mathcal{T}(x)}(dy)\right] d\nu^{\text{o}}(x).
\end{equation*}
Because $\nu^{\text{i}}=\mathcal{T}_*\nu^{\text{o}}$, this gives
\begin{equation}\label{coro_equation_a}
\int_{\mathcal{N}^\text{o}} f(x)\, d\nu^\text{o}(x) =\int_{\mathcal{N}^{\text{i}}}\left[\int_{\mathcal{N}_{\pi(z)}^{\text{o}}} f(y) S_{z}(dy)\right] d\nu^{\text{i}}(z).
\end{equation}
Therefore
\begin{equation*}
d\nu^\text{o}(y)= \int_{\mathcal{N}^{\text{i}}} d\nu^{\text{i}}(z) S_z(dy)=\int_{\mathcal{N}^{\text{i}}}d\zeta_{\text{stat}}(z,y).
\end{equation*}
Thus the outgoing stationary distribution is the incoming stationary distribution acted upon by the scattering kernel.  We can now rewrite
the difference of the two integrals in (\ref{main_entropy_equation}) as follows.  The second term is, by (\ref{coro_equation_a}),
\begin{equation*}
\int_{\mathcal{N}^\text{o}}\frac{E(y)}{\kappa T(\pi(y))} \, d\nu^{\text{o}}(y)=\int_{\mathcal{N}^{\text{i}}}\left[\int_{\mathcal{N}_{\pi(z)}^{\text{o}}} \frac{E(y)}{\kappa T(\pi(z))} S_z(dy)\right] d\nu^{\text{i}}(z).
\end{equation*}
Subtracting this integral from the first term gives
\begin{equation*}
\begin{aligned}
e_p&=\int_{\mathcal{N}^{\text{i}}}\frac1{\kappa T(\pi(z))}\left[E(z)-\int_{\mathcal{N}_{\pi(z)}^{\text{o}}}  E(y) S_z(dy)\right] d\nu^{\text{i}}(z)\\
&=
\int_{\mathcal{N}^{\text{i}}} \int_{\mathcal{N}^{\text{o}}_{\pi(z)}} \frac{E(z)-E(y)}{\kappa T(\pi(z))} S_z(dy)\, d\nu^{\text{i}}(z),
\end{aligned}
\end{equation*}
which is what was to be proved.
\end{proof}

 As will be seen in examples, Equation (\ref{main_entropy_equation}) contains, in our setting, a stochastic version of   Clausius's second law of thermodynamics:
 heat is transferred from a hot source to a cold sink. This equation also makes it clear that the central task in obtaining explicit values for $e_p$
 lies in the determination of the stationary distribution $\nu$. Much of what is done in the rest of the paper is concerned with this issue for a
 class of random scattering operators which we call {\em generalized Maxwell-Smoluchowski} models. 
   
   \section{Entropy production in open compartments}\label{sec:entropy-open-compartments}
   Let us obtain  the total entropy produced during the residence time of a particle in an open compartment; that is, a compartment from which it is possible to exit with
   positive probability. 
   
 \subsection{Compartment operators}
 We have   defined  random  scattering wall operators $S_q$ at points $q\in \mathcal{W}$.    From these we can obtain random scattering  {\em compartment} operators, as explained in this section. Let $\mathcal{W}$ represent here the union of wall pieces enclosing a compartment. Let
 $\mathbf{n}$ denote the unit normal vector field on $\mathcal{W}$ pointing out. Let $\mathcal{W}^{\epsilon}$, $\epsilon\in \{-,+\}$,
 be  the half-space bundle   over $\mathcal{W}$ consisting of all $(q,v)$ such that $q$ lies in a wall piece $W$ and $v\in W_q^\epsilon$.  
 Let $\nu^\text{i}$ denote the (incoming) probability measure on $\mathcal{W}^-$, representing
the probability distribution  of states of a particle impinging on the wall from outside the compartment. It is natural to assume  that,
 once a particle enters the compartment,
its probability of  becoming trapped inside and never leaving is $0$. For many concrete scattering mechanisms this is a consequence of Poincar\'e recurrence. We say, in this case, that the compartment has {\em zero trapping probability}.
The probability measure $\nu^\text{o}$  on  $\mathcal{W}^+$ of states of the particle as it is emitted from the compartment is then a function of $\nu^{\text{i}}$ given by $\nu^{\text{o}}=\nu^\text{i} S$, where $S$, the random scattering compartment operator, is a transition kernel from the incoming states in $\mathcal{W}^-$ to the outgoing states $\mathcal{W}^+$. Thus, for
each incoming state $(q,v)\in \mathcal{W}^-$ and Borel set $U\subset \mathcal{W}^+$, $S(U|q,v)$ is the probability that
the particle which impinges on the surface at $q$ with velocity $v$ eventually leaves (either by immediate reflection or after reemerging from inside the compartment) with outgoing state in $U$. 
Let us write  the compartment operator as $S_{\text{c}}$   while we continue to denote the wall operator by $S$.

As before, we write $S^{\epsilon_1 \epsilon_2}$ for the (wall) block transition kernel taking measures on $\mathcal{W}^{\epsilon_1}$ to measures on $\mathcal{W}^{\epsilon_2}$. For the next proposition (Proposition \ref{proposition:S_diffuse}) we need $A:=(\mathcal{T}^*S^{+-})^2$ to be strictly contracting in the total variation norm. 
The following example should help to clarify the meaning of this assumption. 

\subsubsection{An example}
 Consider a  compartment enclosed by two  parallel walls, $W_1, W_2$,  with
the generalized Maxwell-Smoluchowski random scattering operator, such as the left compartment of Figure \ref{one_compartment_simple} further below.
 Let $S_1$, $S_2$ be the respective random operators with parameters
\begin{equation*}p=p_1=p_2>0,\ \  \alpha=\alpha_1=\alpha_2>0, \ \ C=C_1=C_2\geq 0, \ \ T=T_1=T_2>0.\end{equation*}
These parameters are constant on each wall, so that the position variable $q$ can be replaced with the discrete index $i\in \{1,2\}$. We imagine the walls aligned vertically and call $W_1$ the left wall and $W_2$ the right wall. 
 The return map $\mathcal{T}$ is simply the identity.
 Let   $v\in W_{2}^-$ and set $\nu^-=\delta_{(2,v)}$.
 Then, for any Borel set  $U\subseteq W_1^-$, 
 \begin{equation*}
 \left(\nu^-\mathcal{T}^*S^{+-}\right)(U) =S^{+-}(U|1,v) = (1-p\mathbbm{1}_{\mathsf{T}_1^+}(v))\left[\alpha \mu_{T}(U) +(1-\alpha)\delta_{\overline{v}}(U)\right].
 \end{equation*}
 If the normal component of $v$ is not sufficiently large for the particle to overcome the potential barrier, so that 
 $\mathbbm{1}_{\mathsf{T}_1^+}(v)=0$, then 
 \begin{equation*} \left(\nu^-\mathcal{T}^*S^{+-}\right)(W_1^-)=1\end{equation*}
  so that  $\mathcal{T}^*S^{+-}$ is not 
 strictly contracting. But for the squared operator we have, by a straightforward calculation,
 \begin{equation*}
\left[\nu^-\left(\mathcal{T}^* S^{+-}\right)^2\right](W_1^-)\leq 1-\alpha p(2-p)e^{-\frac{C}{\kappa T}}.
 \end{equation*}
Thus the square of $\mathcal{T}^*S^{+-}$ is strictly contracting when $\alpha$ and $p$ are positive. This is  natural:
for the particle to leave the compartment, the probability $p$ of hitting a pore should be positive
and  some thermal accommodation with the wall is necessary
  for the particle   to overcome the potential barrier, thus $\alpha>0$.

\subsubsection{Compartment operator from wall operator}\label{compartment_proposition}
The random scattering compartment operator $S_\text{c}$ is obtained from the random scattering wall operator $S$ as follows. 
\begin{proposition}[Compartment operator from wall operator]\label{proposition:S_diffuse}
Let $S$ be the random scattering  operator associated to the boundary wall of a bounded compartment. Suppose the block $\left(\mathcal{T}^*S^{+-}\right)^2$ is
strictly contracting, as defined at the beginning of this section. Then the compartment scattering operator $S_{\text{c}}$ is given by
\begin{equation*}
S_{\text{c}}=S^{-+} + S^{--}\left(I-\mathcal{T}^*S^{+-}\right)^{-1} \mathcal{T}^* S^{++},
\end{equation*}
in which
\begin{equation*}
\left(I-\mathcal{T}^*S^{+-}\right)^{-1}=I+\mathcal{T}^*S^{+-} + \left(\mathcal{T}^*S^{+-}\right)^2+\cdots
\end{equation*}
is absolutely convergent in the total variation norm.
\end{proposition}
\begin{proof}
Let $\nu^\text{i}$ be the initial probability distribution of a particle impinging on the outside of the compartment's boundary wall.
We may write this initial distribution as $(\nu^\text{i}, 0)$ to indicate that only the component in  $\mathcal{W}^-$ is nonzero.
 Upon collision, 
the distribution of states splits into a transmitted and a reflected distribution, $(\nu^{\text{i}}S^{--},\nu^{\text{i}}S^{-+})$. The transmitted distribution
is transported back to the wall from the inside under the map $\mathcal{T}$, namely, $\nu^{\text{i}}S^{--}\mathcal{T}^*$, and on a second scattering event it splits into reflected,
$\nu^{\text{i}}S^{--}\mathcal{T}^*S^{+-}$,
 and transmitted, $\nu^{\text{i}}S^{--}\mathcal{T}^*S^{++}$, distributions. Up to this point, the total probability distribution is
 \begin{equation*}
 \left(\nu^{\text{i}}S^{--}\mathcal{T}^*S^{+-}, \nu^{\text{i}}S^{-+}+\nu^{\text{i}}S^{--}\mathcal{T}^*S^{++}\right),
 \end{equation*}
where the  term on the right-hand side of this ordered pair is the total probability (up to $2$ scattering events) for
states emitted from the wall outward and the left-hand side   is the residual probability that the particle remains inside. 
Inductively, let the probability of states by the time of the $n$-th scattering event be
\begin{equation*}
\begin{aligned}
(\nu^{(n)}_-, \nu^{(n)}_+)&=\left(\nu^-S^{--}(\mathcal{T}^*S^{+-})^n,\right.\\
&\ \ \left. \nu^-S^{-+} + \nu^-S^{--}\mathcal{T}^*S^{++}+ \cdots + \nu^-S^{--} \left(\mathcal{T}^*S^{+-}\right)^{n-1}\mathcal{T}^*S^{++}\right).
\end{aligned}
\end{equation*}
On the $n+1$-st scattering event, $\nu^-S^{--}\left(\mathcal{T}^*S^{+-}\right)^n\mathcal{T}^*$ again splits into transmitted,
$\nu^-S^{--}\left(\mathcal{T}^*S^{+-}\right)^n\mathcal{T}^*S^{++}$, and reflected, $\nu^-S^{--}\left(\mathcal{T}^*S^{+-}\right)^n\mathcal{T}^*S^{+-}$, distributions, where the former is added to the sum of emitted probabilities and the latter is the residual remaining inside the compartment.
Due to the strictly contracting assumption, the residual distribution converges to zero and the emitted distribution
is given by the sum
\begin{equation*}
\nu^\text{o}=\nu^{\text{i}}\left[S^{-+} + S^{--}\left(I+\mathcal{T}^*S^{+-} + \left(\mathcal{T}^*S^{+-}\right)^2+\cdots\right)\mathcal{T}^*S^{++}\right].
\end{equation*}
This is what was to be proved.
\end{proof}

It is interesting to reinterpret this proposition from a stationary perspective. From this perspective, we think of $\nu^{\text{i}}$ as
a stationary distribution indicating a flow of particles impinging on the wall of the compartment from the outside. This generates a stationary flow
of emitted particles with velocities distributed according to $\nu^{\text{o}}=\nu^{\text{i}}S_{\text{c}}$. To obtain the latter, let $\eta^{-}$ and $\eta^{+}$ be stationary distributions of  velocities of emitted and absorbed particles on the inner side of the wall. In the stationary regime, the emitted particles are reabsorbed after translation under $\mathcal{T}$,  so  $\eta^+=\mathcal{T}_*\eta^-$.
This gives the stationary equation
\begin{equation*}
(\eta^{-},\nu^{\text{o}}) =(\nu^{\text{i}},\eta^{+})\left(\begin{array}{cc}S^{--} & S^{-+} \\[3pt]S^{+-} & S^{++}\end{array}\right).
\end{equation*} 
This gives the equations
\begin{equation*}
\begin{aligned}
\eta^{-}&=\nu^{\text{i}} S^{--} + \eta^{-}\mathcal{T}^*S^{+-}\\
\nu^{\text{o}}&=\nu^{\text{i}} S^{-+}+ \eta^{-}\mathcal{T}^*S^{++}.
\end{aligned}
\end{equation*}
These equations are easily solved for $\nu^{\text{o}}$ as a function of $\nu^{\text{i}}$ by eliminating $\eta^{-}$, giving the same expression as in 
Proposition \ref{proposition:S_diffuse}:
\begin{equation*}
\nu^{\text{o}}=\nu^{\text{i}}\left[S^{-+} + S^{--}\left(I-\mathcal{T}^*S^{+-}\right)^{-1} \mathcal{T}^* S^{++}\right]. 
\end{equation*}

\subsubsection{Composition of walls and the Redheffer star-product}\label{redheffer}
We note, in passing, a different interpretation of the main  algebraic observation of the previous subsection,
having to do with  constructing composite walls  by layering 
several walls next to each other. 
The same procedure used above to obtain 
compartment operators also yields the operators for such  composite walls. These  operators are obtained from those of the individual wall layers
by an algebraic operation that, in its simplest form, corresponds to the  {\em Redheffer star-product}, well-known in scattering theory.
 
 \begin{figure}[ht]
\begin{center}
\includegraphics[width=3.0in]{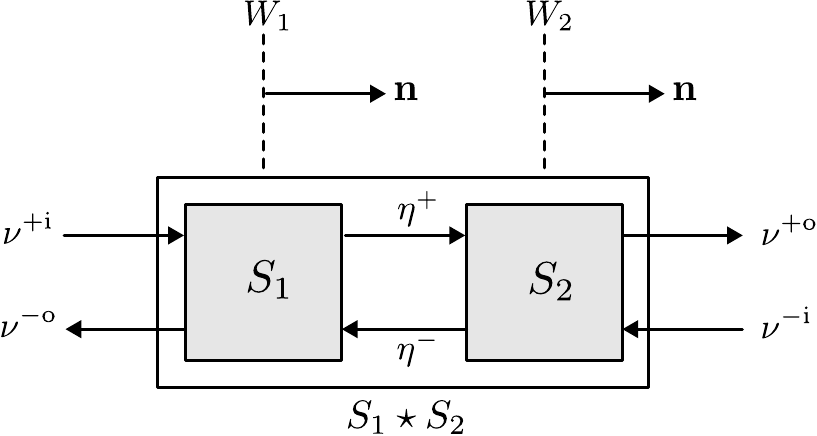}
\caption{{\small  The Redheffer star product of the random scattering operators for two parallel walls.  }}
\label{S1_star_S2}
\end{center}
\end{figure}

It seems  more natural here to use  the same
normal vector for all wall layers.   See the diagram of 
  Figure \ref{S1_star_S2}  for the case of two layers.  
Denoting the random scattering operators of walls $W_1$, $W_2$ by $S_1$, $S_2$, respectively,
the compartment scattering operator of Proposition \ref{proposition:S_diffuse} then takes the form
\begin{equation*}
\left(\nu^{-\text{o}},\nu^{+\text{o}}\right) =\left(\nu^{-\text{i}},\nu^{+\text{i}}\right) (S_1\star S_2),
\end{equation*}
in which
\begin{equation*}
S_1\star S_2=\left(\begin{array}{cc}S_{12}^{--} & S^{-+}_{12}\\[3pt]S_{12}^{+-}&S_{12}^{++}\end{array}\right)
\end{equation*}
with block operators
\begin{equation*}%\label{redheffer_blocks}
\begin{aligned}
S^{--}_{12}&=S_2^{--}\left(I-S_1^{-+}S_2^{+-}\right)^{-1}S_1^{--}\\
S^{++}_{12}&= S_1^{++}\left(I-S_2^{+-}S_1^{-+}\right)^{-1}S_2^{++}\\
S^{-+}_{12}&=S_2^{--}S_1^{-+}\left(I-S_2^{+-}S_1^{-+}\right)^{-1} S_2^{++} +S_2^{-+} \\
S^{+-}_{12}&=S_1^{++}S_2^{+-}\left(I-S_1^{-+}S_2^{+-}\right)^{-1} S_1^{--} +S_1^{+-}.
\end{aligned}
\end{equation*}

The binary operation $S_1\star S_2$ is   the {Redheffer star-product}. It can be shown to be associative. Thus
the random   scattering operator for the layering of several walls (or the concatenation of  compartments separated by vertical walls)
is given by the  Redheffer products of the wall operators in the given order (from left to right, according to our conventions).

\subsubsection{Compartment statistics}\label{compartment_statistics}
We consider here a single compartment $M$ with boundary wall $\mathcal{W}$ and 
the random dynamical system (or Markov chain), $X_0, X_1, \dots$, defined in Section \ref{rds}, up to the exit time-step
\begin{equation*}
N_{\mathcal{W}^-}:= \text{min}\{j\geq 0: X_j\notin \mathcal{W}^-\}.
\end{equation*}
We are orienting the boundary of $M$ with the unit  normal vector field $\mathbf{n}$ pointing out.  Thus the chain remains in the compartment 
for as long as $X_i$ lies in $\mathcal{W}^-$.
Let $Q:=\mathcal{T}S^{+-}$. (Recall that $\mathcal{T}$ acts on measures according to  $(\nu \mathcal{T})(f)=\nu(f\circ \mathcal{T})$.) This is  also written as the substochastic kernel
$Q(U|x)=B(U\cap \mathcal{W}^-|x)$, where $B$ is the generator of the random dynamical system given in Definition \ref{one_step_generator}.
In terms of the random scattering operators, we may write, for any test function $f$,
\begin{equation*}
(Qf)(q,v)=Q(f|(q,v))=S^{+-}_{\mathcal{T}(q,v)}\left(f|_{W_{\mathcal{T}(q,v)}^-}\right),
\end{equation*}
where $v\in W_q^-$.

 We assume, as in Subsection  \ref{compartment_proposition}, Proposition \ref{proposition:S_diffuse},  that $Q^2$ is strictly contracting.
Thus $Q$ defines the  substochastic kernel obtained by restricting   the transition kernel on the random dynamical system to the transient set $\mathcal{W}^-$.  
Our concern here is to obtain expressions for the expected value of  such quantities as the  exit time from $\mathcal{W}^-$ and the amount of entropy produced before leaving the compartment (that is, when $X_n\in \mathcal{W}^+$). As was the case in Proposition \ref{proposition:S_diffuse}, the resolvent $(I-Q)^{-1}$ plays a key role.

\begin{proposition}\label{pattern_proposition}
Let $f(x_1, \dots, x_k)$ be a bounded Borel measurable  function with support on the $k$-fold Cartesian product of the transient set $\mathcal{W}^-$. Consider the Markov chain $X_0, X_1, \dots$
with   $X_0=x=(q,v)\in \mathcal{W}^-$. We suppose $Q^2$ is strictly contracting. Denote by $Q_if$ the action of $Q$ on the variable $x_i$. Then
\begin{equation*}
u(x):=\mathbb{E}_x\left[\sum_{n=0}^{\infty} f(X_{n}, X_{n+1}, \dots, X_{n+k-1})\right]
=\left[\left(I-Q\right)^{-1}\left(Q_2\dots Q_kf\right)\right](x).
\end{equation*}
Thus $u(x)$ satisfies the Poisson equation
\begin{equation*}
u(x)-\int_{\mathcal{W}^-} u(y)Q(dy|x)=\left(Q_2\dots Q_kf\right)(x)
\end{equation*}
for $x\in \mathcal{W}^-$. 
\end{proposition}
\begin{proof}
This follows from a standard application of the Markovian nature of the chain and properties of conditional expectation. 
\end{proof}

\begin{corollary}\label{corollary_residence_time}
The expected number of steps, $u(x):=\mathbb{E}_x\left[N_{\mathcal{W}_-}\right]$,   till the exit  from the transient set $\mathcal{W}^-$ satisfies
\begin{equation*}
u(x)-\int_{\mathcal{W}^-} u(y)Q(dy|x)=1,
\end{equation*}
for all $x\in \mathcal{W}^-$, so that
\begin{equation*}
\mathbb{E}_x\left[N_{\mathcal{W}^-}\right]=\left[(I-Q)^{-1}\mathbf{1}\right](x).
\end{equation*}
 If  $\phi(x)$ is an integrable function of the particle state at collisions, such as   the time duration of the free-flight segment in the compartment $M$ with initial state $x$, then
\begin{equation*}\mathbb{E}_x\left[\sum_{n=0}^\infty \phi(X_n)\right]=\left[(I-Q)^{-1}\phi\right](x).\end{equation*}
If $\phi$ is the time of free-flight, the above  is
the expected exit time from $\mathcal{W}^-$.
\end{corollary}
\begin{proof}
The expression for the expected number of steps is an immediate consequence of Proposition \ref{pattern_proposition}.
The claim about the expected exit time needs to consider the possibility that $\phi$ is not bounded. This issue is dealt with by a standard application of monotone convergence.
\end{proof}

For the next proposition, let $T(x)=T(q)$ be the temperature at $x=(q,v)$, $q\in \mathcal{W}$. Recall that $E_0(x)=\frac12 m |v|^2_q$ is the kinetic energy
of state $x$ and for $x=(q,v)\in \mathcal{W}^+$, $y=(q,w)\in \mathcal{W}^-$, write
\begin{equation*}
\sigma(x,y):=-\frac{E_0(y)-E_0(x)}{\kappa T(x)}.
\end{equation*}
This is the entropy produced by the transfer of kinetic energy to the wall (imagined  to be in contact with a heat bath) at temperature $T(x)$.
We will later approach the issue of entropy production in a more fundamental way, but take this expression as a definition for now.
Then the entropy produced during the particle's stay in compartment $M$ is 
\begin{equation}\label{S_formula}
S=\sum_{n=0}^{N_{\mathcal{W}^-}-1}\sigma(\mathcal{T}(X_n),X_{n+1})=\sum_{n=0}^\infty\sigma(\mathcal{T}(X_n),X_{n+1})\mathbbm{1}_{\mathcal{W}^-\cup\mathcal{W}^+}(X_{n+1}).
\end{equation}
In the above,  collisions at states selected by the indicator function include  the one when the particle leaves the compartment. 
Define for $x\in \mathcal{W}^-$ the functions
\begin{equation}\label{definition_h}
h(x)=\int_{\mathcal{W}^-}\sigma(\mathcal{T}(x),y) Q(dy|x), \ \ g(x)=\int_{\mathcal{W}^+_{\pi(\mathcal{T}(x))}} \sigma(\mathcal{T}(x),y) S^{++}(dy|\mathcal{T}(x)).
\end{equation}
Here $h(x)$ is the expected entropy produced by a collision that stays inside the
compartment, while $g(x)$ is the expected entropy produced by the terminal
collision that transmits the particle out of the compartment.

\begin{proposition}\label{total_residence_entropy}
The expected total entropy produced during the residence time in compartment $M$ for the process started at $x$ is
\begin{equation*}
\mathbb{E}_x[S]=\left[(I-Q)^{-1}(h+g)\right](x).
\end{equation*}
\end{proposition}
\begin{proof}
Let $S_{\text{int}}$ be entropy produced by the collisions sending states back into $\mathcal{W}^-$:
\begin{equation*}
S_{\text{int}}=\sum_{n=0}^{N_{\mathcal{W}^-}-2}\sigma(\mathcal{T}(X_n),X_{n+1})=\sum_{n=0}^\infty\sigma(\mathcal{T}(X_n),X_{n+1})\mathbbm{1}_{\mathcal{W}^-}(X_{n+1}).
\end{equation*}
Taking the expectation of   $S$     (Equation (\ref{S_formula})), 
\begin{equation*}
\begin{aligned}
\mathbb{E}_x[S_{\text{int}}]&=\mathbb{E}_x\left[\sum_{n=0}^\infty \sigma(\mathcal{T}(X_n),X_{n+1})\mathbbm{1}_{\mathcal{W}^-}(X_{n+1})\right]\\
&=\sum_{n=0}^\infty \mathbb{E}_x\left[\mathbb{E}\left[\sigma(\mathcal{T}(X_n),X_{n+1})\mathbbm{1}_{\mathcal{W}^-}(X_{n+1})|X_n\right]\right].\end{aligned}
\end{equation*}
Given $X_n=x$ in the transient set $\mathcal{W}^-$, the inner conditional expectation  is
\begin{equation*}
\mathbb{E}\left[\sigma(\mathcal{T}(x),X_{n+1})\mathbbm{1}_{\mathcal{W}^-}(X_{n+1})\right]=\int_{\mathcal{W}^-}\sigma(\mathcal{T}(x),y) Q(dy|x)=h(x).
\end{equation*}
Moreover, if $X_n$ is not in $\mathcal{W}^-$, the contribution to the expectation is zero. Hence
\begin{equation*}
\mathbb{E}_x\left[S_{\text{int}}\right]=\sum_{n=0}^\infty \mathbb{E}_x\left[h(X_n)\mathbbm{1}_{\mathcal{W}^-}(X_n)\right]=\sum_{n=0}^\infty \int_{\mathcal{W}^-} h(y)Q(dy|x) = \sum_{n=0}^\infty (Q^nh)(x).
\end{equation*}
The last expression can be written using the resolvent applied to $h$. The contribution of the terminal  collision, before the particle leaves 
the compartment,  is
\begin{equation*}
\mathbb{E}_x\left[\sigma\left(\mathcal{T}(X_{N_{\mathcal{W}^-}-1}), X_{N_{\mathcal{W}^-}}\right)\right]=
\mathbb{E}_x\left[
\mathbb{E}\left[
\sigma\left(\mathcal{T}(X_{N_{\mathcal{W}^-}-1}), X_{N_{\mathcal{W}^-}}\right)
\left| X_{N_{\mathcal{W}^-}-1}\right.
\right]
\right].
\end{equation*}
Given $X_{N_{\mathcal{W}^-}-1}=y$, the   conditional expectation on the left-hand side of the above  equation is $g(y)$.
Similar to what was done for $h$, we obtain the contribution $\left[(I-Q)^{-1}g\right](x)$ by the terminal collision.
\end{proof}

We finally consider the exit probability distribution. Let $A$ be a measurable subset of $\mathcal{W}^+$ and denote by 
$P_{\text{exit}}(A|x)$ the conditional  probability that
$X_{N_{\mathcal{W}^-}}\in A$ given that $X_0=x\in \mathcal{W}^-$. Then
\begin{equation*}
P_{\text{exit}}(A|x)=S^{++}(A|\mathcal{T}(x))+\int_{\mathcal{W}^-}P_{\text{exit}}(A|y) Q(dy|x).
\end{equation*}
This may also be written in operator form as
\begin{equation*}
P_{\text{exit}} =\mathcal{T}S^{++} +Q P_{\text{exit}}.
\end{equation*}
Therefore
\begin{equation}\label{exit_probability}
P_{\text{exit}}=(I-Q)^{-1} \mathcal{T}S^{++}.
\end{equation}
We register this fact in the following form.
\begin{proposition}
Let $i, j$ indicate two adjacent compartments separated by wall $W_{ij}$. We orient the boundary of compartment $i$ so that the unit normal vector field $\mathbf{n}$ points out. The random scattering operators   $S^{+-}$, $S^{++}$, $Q=\mathcal{T}S^{+-}$, and $P_{\text{exit}}$
refer  to compartment $i$. 
 Let $x$ be an initial state pointing into compartment $i$. Then the probability that the random process will transition into compartment $j$ when it exits compartment $i$ is
 \begin{equation*}
 p_{ij}(x)=P_{\text{\em exit}}(W_{ij}^{+}|x),
 \end{equation*}
 where $P_{\text{exit}}$ is given in Equation (\ref{exit_probability}). Here $W_{ij}^+$ is the set of outgoing  states crossing $W_{ij}$ into compartment $j$.
\end{proposition}

The notation \begin{equation*}%\label{exit_probability_ij}
p_{ij}(dy|x):=P_{\text{exit}}(dy|x)
\end{equation*}
 for the exit measure on $W_{ij}$ will  be used in Subsection \ref{entropy_composition_section}.

\subsection{Composition of compartments}
Random compartment scattering operators can be composed.
Consider a pair of  regions, $M_1, M_2$,  enclosed by walls $\mathcal{W}_1$ and $\mathcal{W}_1$, respectively. Each region may be subdivided into
compartments. Suppose that $M_1$ and $M_2$ share a piece of boundary, $\mathcal{W}=\mathcal{W}_1\cap\mathcal{W}_2$, and denote by
$\mathcal{W}'_i$ the complement  of $\mathcal{W}$ in $\mathcal{W}_i$.
We take here the perspective that measures correspond to stationary flows of particles impinging on, or being emitted from, walls. 
Let $\nu_j^{\text{i}}$ and $\nu_{j}^{\text{o}}$ be the incoming and outgoing measures on $\mathcal{W}_j^{'}$. Let $\eta$  be
the measure associated to the crossing from $M_1$ to $M_2$ through $\mathcal{W}$ and $\eta'$, similarly, from $M_2$ to $M_1$. 

 \begin{figure}[ht]
\begin{center}
\includegraphics[width=4.0in]{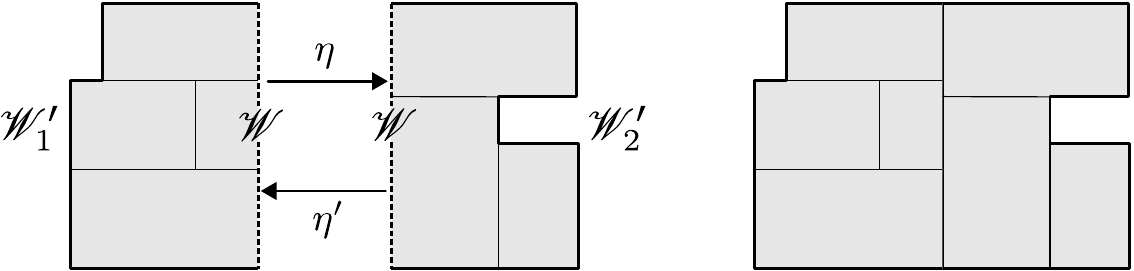}\ \ 
\caption{{\small  Composition of   random scattering operators of two   regions with permeable boundaries sharing a common boundary piece $\mathcal{W}$. The boundary pieces $\mathcal{W}_j'$ are shown in bold solid lines and the intersection of the boundary walls in bold dashed lines.}}
\label{compartment_composition}
\end{center}
\end{figure} 

The same ideas used in Section \ref{compartment_proposition} are employed here. Let $S_1, S_2$ be
the compartment operators. We indicate the block operators by $S_j^{\epsilon_1\epsilon_2}$ in which the signs $\epsilon_k$
are relative to the unit normal vector fields along the boundary walls, pointing out.
The incoming measure after composition
is   $\nu_1^{\text{i}}+\nu_2^{\text{i}}$ and the outgoing is $\nu_1^{\text{o}}+\nu_2^{\text{o}}$. The outgoing measure is a function of the incoming obtained from the system
\begin{equation*}
\begin{aligned}
\left(\nu_1^{\text{o}}, \nu_2^{\text{o}}\right)&=\left(\nu_1^\text{i} S_1^{+-}+\eta' S_1^{--},  \nu_2^{\text{i}}S_2^{-+}+\eta S_2^{++}\right)\\ 
\left(\eta,\eta'\right)&=\left( \nu_1^\text{i} S_1^{++}+\eta' S_1^{-+},\nu_2^{\text{i}}S_2^{--}+\eta S_2^{+-}\right).
\end{aligned}
\end{equation*}
Eliminating $\eta, \eta'$ gives
\begin{equation*}
\begin{aligned}
\nu_1^\text{o}&=\nu_1^\text{i}S_1^{+-}+\nu_2^{i} S_2^{--}S_1^{--} + \left(\nu_1^\text{i}S_1^{++}+\nu^\text{i}_2S_2^{--}S_1^{-+}\right)\left(I-S_2^{+-}S_1^{-+}\right)^{-1}S_2^{+-}S_1^{--} \\
\nu_2^\text{o}&=\nu_2^\text{i}S_2^{-+}+\nu_1^{i} S_1^{++}S_2^{++} + \left(\nu_1^\text{i}S_1^{++}+\nu^{\text{i}}_2S_2^{--}S_1^{-+}\right)\left(I-S_2^{+-}S_1^{-+}\right)^{-1}S_1^{-+}S_2^{++}.
\end{aligned}
\end{equation*}
\subsection{Entropy production composition}\label{entropy_composition_section}
The entropy production rate of a closed compartmented system with finitely many compartments, $M_1, \dots, M_m$, can now be expressed
in terms of the contributions of   each compartment, using the already obtained expectations for total entropy, number of collisions, and escape probability distributions during sojourns in compartments.

A sojourn of the particle in a compartment $M_i$ begins when the particle enters $M_i$ through a separating wall, with initial state $x\in \mathcal{W}_i^-$ (the bundle of incoming states through the boundary $\mathcal{W}_i$ of $M_i$),
and ends when it exits $M_i$ into an  an adjacent compartment.  Inside $M_i$ the process is a Markov chain with substochastic kernel $Q_i=\mathcal{T}S^{+-}_i$ (restricted to the transient set $\mathcal{W}_i^-$).  We have obtained in Subsection \ref{compartment_statistics}
the main statistics of this sojourn:
\begin{enumerate}
\item The expected total entropy produced during the visit, starting from $x$:
\begin{equation*}
\mathbb{E}_x[S_i]=\left[(I-Q_i)^{-1}h_i\right](x);
\end{equation*}
\item The expected number of wall collisions during the visit:
\begin{equation*}
\mathbb{E}_x[N_i]=\left[(I-Q_i)^{-1}1\right](x);
\end{equation*}
\item The probability distribution of the exit state $y$ into an adjacent compartment $M_j$:
\begin{equation*}
p_{ij}(dy|x)=P_{\text{exit}}(dy|x)=\left[(I-Q_i)^{-1} \mathcal{T}S_i^{++}\right](dy|x).
\end{equation*}
\end{enumerate}
The sequence of entrance states into the successive compartments forms a new Markov chain, whose state space is the disjoint union of the entrance boundaries $\mathcal{W}_1^-, \dots, \mathcal{W}_m^-$.
This chain has a  transition kernel  given by the exit distribution:
\begin{equation*}
\widetilde{P}(dy|x) =p_{ij}(dy|x), \ \ x\in \mathcal{W}_i^-, y\in \mathcal{W}_j^-.
\end{equation*}
We make the assumption, to be verified in   examples, that there exists a unique stationary probability measure $\eta$ for this entrance chain. So $\eta=\eta\tilde{P}$. Implicit in this assumption is the irreducibility and recurrence of the full random flight system.  We call $\eta$ the 
{\em stationary entrance probability measure} and $\tilde{P}$ the {\em compartment transition kernel}.
Let $\eta_i$ denote the restriction of $\eta$ to $\mathcal{W}_i^-$.

\begin{theorem}\label{entropy_multicompartment} Consider the random flight  in a closed compartmented system with compartments $M_1, \dots, M_m$.
Under the assumption of  existence and uniqueness of the stationary entrance probability measure $\eta$ for the   compartment transition kernel 
$\widetilde{P}$,  the entropy production rate   can be expressed in terms of
 the  compartment contributions as
\begin{equation}\label{theorem_entropy_compartments}
e_p= \frac{\displaystyle \sum_{i=1}^m \int_{\mathcal{W}_i^-} \mathbb{E}_x[S_i] \, d\eta_i(x)}
          {\displaystyle \sum_{i=1}^m \int_{\mathcal{W}_i^-} \mathbb{E}_x[N_i] \, d\eta_i(x)}=\frac{\displaystyle \sum_{i=1}^m \int_{\mathcal{W}_i^-} \bigl[(I - Q_i)^{-1} (h_i+g_i)\bigr](x) \, d\eta_i(x)}
          {\displaystyle \sum_{i=1}^m \int_{\mathcal{W}_i^-} \bigl[(I - Q_i)^{-1}{1}\bigr](x) \, d\eta_i(x)}.
\end{equation}
\end{theorem}
\begin{proof}
Given the work already done in Subsection \ref{compartment_statistics}, Equation (\ref{theorem_entropy_compartments})
is now a direct consequence of the renewal-reward theorem for Markov chains.
\end{proof}

\subsection{Probability fluxes between compartments}
We introduced in the previous section the {\em entrance Markov chain}, which has transition kernel $\widetilde{P}$ (with the associated transition measures $p_{ij}(dy|x)$) and stationary {\em entrance probability measure} $\eta$. It makes sense to define the stationary {\em probability flux} $\Phi_{ij}$ from compartment $i$ to compartment $j$
as follows. Denote by  $\mathcal{W}_i^-$ the bundle of half-spaces over the boundary of compartment $i$ consisting of velocity vectors pointing
in, by  $\mathcal{W}^-$  the disjoint union of the $\mathcal{W}_i^-$, and by $\mathcal{W}_{ij}^-$ the subset of $\mathcal{W}_j^-$ given by the  intersection $\mathcal{W}_i^+\cap\mathcal{W}_j^-$. Write $\eta_i:=\eta|_{\mathcal{W}_i^-}$. Then
\begin{equation}\label{probability_flux_definition}
\Phi_{ij}:=\int_{\mathcal{W}_i^-}\eta_i(dx) p_{ij}\left(\mathcal{W}_{ij}^-| x\right) =\int_{\mathcal{W}_i^-}\eta(dx) p_{ij}(\mathcal{W}^-|x).
\end{equation}
It is an immediate consequence of stationarity that
\begin{equation*}%\label{probability_flux_balance}
\sum_{j}\Phi_{ij} =\sum_j\Phi_{ji}.
\end{equation*}
Thus the total probability flux out of compartment $i$ equals the total flux into it.

As an example, consider $q\geq 3$ compartments, labeled $0, 1, \dots, ,q-1$, arranged in cyclic order, so  that
compartment $i$ only has intersecting  boundary with compartments $i-1$ and $i+1$ modulo $q$.  Then, for each $i$,
\begin{equation*}
\Phi_{i\, i+1}+\Phi_{i\, i-1}= \Phi_{i+1\, i}+\Phi_{i-1\, i}.
\end{equation*}
Thus the probability net flux through the common boundary between any two adjacent compartments is the same for all $i$:
\begin{equation}\label{probability_flux_balance_cycle}
J:=\Phi_{i-1\, i}-\Phi_{i\, i-1}=\Phi_{i\, i+1}-\Phi_{i+1\, i}.
\end{equation}
A non-zero $J$ indicates particle circulation,  in the stationary regime, around the cyclic system of compartments. The circulation per unit time is
obtained by dividing this quantity by the expected sojourn time inside a compartment (the expected time between two consecutive transitions). 

\section{Examples of compartmented systems}\label{sec:examples}
We consider a few examples of systems with a small number of compartments ($1$, $2$, and $3$) and the generalized Maxwell-Smoluchowski random scattering operators. In the simple models considered here,
position coordinate  is reduced to a discrete wall  index while velocities are tracked in detail. Due to Theorem \ref{entropy_multicompartment},
which reduces entropy production of multicompartment systems to the contribution of individual compartments, we give particular attention to the open single compartment case.

\subsection{Simple one compartment examples}
Let us illustrate the above ideas for the  single compartment example of Figure \ref{one_compartment_simple}. 

 \begin{figure}[ht]
\begin{center}
%\epsfile{file=bundle.eps,scale=0.8}
\includegraphics[width=2.5in]{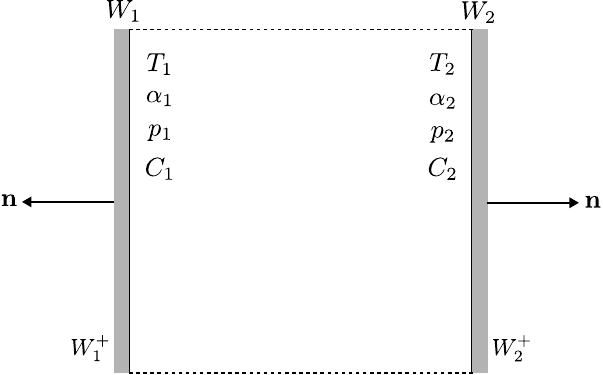}\ \ 
\caption{{\small Single compartment bounded by two plates. It is open if at least one of $p_1$ and $p_2$ is positive.  We use the Maxwell-Smoluchowski model and assume that  wall parameters  are constants along each wall; thus
 position will be indicated simply by the wall index $i=1,2$.  }}
\label{one_compartment_simple}
\end{center}
\end{figure} 

We characterize the two  walls   by the Maxwell-Smoluchowski
model with parameters $p_i, \alpha_i, T_i, C_i$, $i=1,2$.  
Wall orientation is defined by unit normal vectors 
pointing out (Figure \ref{one_compartment_simple}). The notation $\overline{1}=2$ and $\overline{2}=1$ will be used.  Under the natural  identification $W^-_i\cong W_{\overline{i}}^+$
the action of   $\mathcal{T}$ on vectors is the identity map, so $\mathcal{T}(i,v)=(\overline{i},v)$. 
It will be occasionally convenient to identify all half-spaces with 
\begin{equation*}\mathbb{R}^n_+=\{v=(v_1, \dots, v_n)\in \mathbb{R}^n:v_n>0\}.
\end{equation*}
In particular, the specular reflection $\text{ref}(v)=\overline{v}$ of $v$ on the tangent space at  a wall point
may also be regarded as the identity map and   the upper sign may be dropped from the symbols $\mathsf{T}_i^\epsilon$ and $\mathsf{R}_i^\epsilon$. For clarity, this latter identification (manifested in dropping the $\pm$ sign) will be used, initially,  only when referring to the Maxwellian   measures $\mu_i:=\mu_{T_i}$.  We thus have
\begin{equation*}
 \mathcal{N}^{\text{o}}=W_1^-\sqcup W_1^+\sqcup W_2^-\sqcup W_2^+, \ \ 
 W_1^-\cong W_2^+\cong H^+, \ \ W_2^-\cong W_1^+\cong H^-
\end{equation*}
 where $H^+=\mathbb{R}^n_+$ and $H^-=\text{ref}\,H^+$. 

 \begin{figure}[ht]
\begin{center}
%\epsfile{file=bundle.eps,scale=0.8}
\includegraphics[width=3.0in]{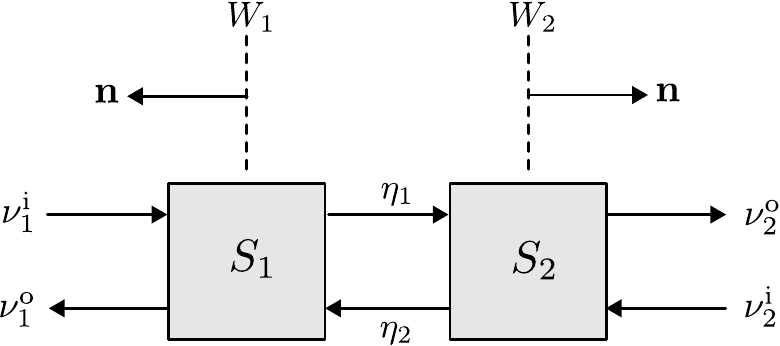}\ \ 
\caption{{\small  Diagram showing the subprobability (stationary) velocity distribution measures.}}
\label{one_compartment_flow_diagram}
\end{center}
\end{figure}

As in Subsection \ref{redheffer} we obtain
\begin{equation}\label{in_out_equation_example}
\begin{aligned}
\nu_1^\text{o}&=\nu_1^{\text{i}}\left[S_1^{-+}+S_1^{--}S_2^{+-}A_{12}S_1^{++}\right] + \nu_2^{\text{i}}S_2^{--}A_{12}S_1^{++}\\
\nu_2^\text{o}&= \nu_1^{\text{i}}S_1^{--}A_{21}S_2^{++}+\nu_2^{\text{i}}\left[S_2^{-+}+S_2^{--}S_1^{+-}A_{21}S_2^{++}\right] ,
\end{aligned}
\end{equation}
in which 
\begin{equation}\label{equation_A12}
A_{12}:=\left(I-S_1^{+-}S_2^{+-}\right)^{-1}, \ \ \ A_{21}=\left(I-S_2^{+-}S_1^{+-}\right)^{-1}.
\end{equation}
We also have
\begin{equation*}
\eta_1=\left(\nu_1^{\text{i}}S_1^{--} +\nu_2^{\text{i}} S_2^{--}S_1^{+-}\right)A_{21}, \ \ \eta_2=\left(\nu_1^{\text{i}}S_1^{--}S_2^{+-} +\nu_2^{\text{i}}S_2^{--}\right)A_{12}.
\end{equation*}

\subsubsection{Open compartment scattering operator}\label{sec:openCompartment}
In this example, we assume $C_1=C_2=0$,  $\alpha_1=\alpha_2=1$, $p_1=p_2=p$, but possibly different temperatures $T_1, T_2$. (The assumptions that $\alpha_i=1$ and $C_i=0$ greatly simplify the  calculations. We will see the effect of having accommodation coefficients 
less than $1$ and positive potential in other examples below.)
Let us obtain the compartment operator $S=S_1\star S_2$. (Keep in mind the difference in orientation of $W_1$ and  $W_2$ compared to
subsection \ref{redheffer}.) 
The operators $S^{\epsilon_1\epsilon_2}_i$ are
\begin{equation*}
(S^{\epsilon\epsilon}_if)(v) = p\mu_i\left(f|_{W_i^\epsilon}\right), \ \ (S^{\epsilon\overline{\epsilon}}_i f)(v)=(1-p)\mu_i\left(f|_{W_i^{\overline{\epsilon}}}\right),
\end{equation*}
thus constant in $v\in W_i^{\epsilon}$.
One easily obtains
\begin{equation*}
A_{12} = I +\frac{(1-p)^2}{1-(1-p)^2}\mu_2, \ \  A_{21} = I +\frac{(1-p)^2}{1-(1-p)^2}\mu_1.
\end{equation*}
Here the Maxwellians $\mu_i$ are viewed as operators (linear functionals) acting on test functions. The action of
$\mu_i$ on measures, from the right, is $\nu \mu_i= r\mu_i$, where $r$ is the total measure, $\nu(1)$,  of (the subprobability measure) $\nu$. 
The explicit form of (\ref{in_out_equation_example}) is
\begin{equation*}%\label{in_out_equation_153}
\begin{aligned}
\nu_1^\text{o}&=\left[\frac{2(1-p)}{2-p}r_1^\text{i}+\frac{p}{2-p}r_2^\text{i}\right]\mu_1\\
\nu_2^{\text{o}}&=\left[\frac{p}{2-p}r_1^{\text{i}} + \frac{2(1-p)}{2-p}r_2^{\text{i}}\right]\mu_2.
\end{aligned}
\end{equation*}
Note that the general relation \begin{equation*}\nu_1^{\text{o}}(1) +\nu_2^{\text{o}}(1)=\nu_1^{\text{i}}(1) +\nu_2^{\text{i}}(1)
\end{equation*} holds and that
\begin{equation*}\lim_{p\rightarrow 0}\nu_j^\text{o}=r_j^{\text{i}}\mu_j, \ \ \lim_{p\rightarrow 1}\nu_j^{\text{o}}=r^{\text{i}}_{\overline{j}}\mu_j.\end{equation*}

The internal measures $\eta_1, \eta_2$ are given by
\begin{equation*}
\begin{aligned}
\eta_1&=\left(\nu_1^{\text{i}}S_1^{--}+\nu_2^{\text{i}} S_2^{--} S_1^{+-} \right)A_{21}=\left( \frac{1}{2-p} r_1^{\text{i}} + \frac{1-p}{2-p}r_2^\text{i}\right) \mu_1\\
\eta_2&=\left(\nu_1^{\text{i}} S_1^{--} S_2^{+-} +\nu_2^{\text{i}}S_2^{--}\right)A_{12}=
\left(\frac{1-p}{2-p}r_1^\text{i} +\frac{1}{2-p} r_2^{\text{i}}\right) \mu_2,
\end{aligned}
\end{equation*}
where $r_j^{\text{i}}=\nu_j^{\text{i}}(1)$.

Let us  find the expected number of collisions, expected total time,  and expected total entropy produced during the residence time inside the open compartment. 
For this, we use the results from Subsection \ref{compartment_statistics} and the operator $Q=\mathcal{T} S^{+-}$. In the present example, 
as $\mathcal{T}(i,v)=\left(\overline{i},v\right)$, (recall that $\overline{1}=2$, $\overline{2}=1$),
\begin{equation}\label{Q_equation}
(Qf)(i,v)= S^{+-}_{\overline{i}}\left(\left.f|_{W^-_{\overline{i}}}\right| v\right).
\end{equation}
Finding the resolvent $(I-Q)^{-1}$ requires solving for $f$ the linear equation $(I-Q)f=g$. The solution is
\begin{equation}\label{resolved_applied_to_g}
\left[(I-Q)^{-1}g\right](i,v)=g(i,v) + \frac{1-p}{p(2-p)}\left[\mu_{\overline{i}}(g) + (1-p)\mu_i(g)\right].
\end{equation}
The expected number of collisions, $N(i,v)$, during the residence time in the compartment, with initial condition $(i,v)$, is
$(I-Q)^{-1}1$. Thus
\begin{equation*}
\mathbb{E}_{(i,v)}[N]=1+\frac{1-p}{p}=\frac1p,
\end{equation*}
which is independent of the initial condition.

The expected total time during residence in the compartment is $(I-Q)^{-1}t$, where $t(i,v)=-{\ell}{v\cdot\mathbf{n}_{i}}$, $v\in W_i^-$.
Here $\ell$ is the distance separating $W_1$ and $W_2$. This expected value of $t$, obtained using Corollary \ref{corollary_residence_time},  is
\begin{equation*}
\mathbb{E}_{(i,v)}[t]=-\frac{\ell}{v\cdot \mathbf{n}_i} +\frac{1-p}{p(2-p)}\ell \pi \left[\sqrt{\frac{m}{2\pi \kappa T_{\overline{i}}}} +(1-p)
\sqrt{\frac{m}{2\pi \kappa T_{{i}}}} \right],
\end{equation*}
in which $E_0(v)=\frac12 m|v|^2$. This value naturally goes to infinity as $p$ approaches $0$, and decreases when the temperatures 
increase.

The expected total entropy produced during residence in the compartment is given in Proposition \ref{total_residence_entropy}. For the present example,
we first need the functions giving the expected entropy produced at one collision:
\begin{equation*}
\begin{aligned}
h(i,v)&=(1-p) \int_{W_{\overline{i}}^+} \frac{E_0(v)-E_0(u)}{\kappa T_{\overline{i}}}\, d\mu_{\overline{i}}(u)\\
&=(1-p)\left[\frac{\frac12m|v|^2}{\kappa T_{\overline{i}}}-\frac2{\pi^{\frac{n-1}{2}}}\int_{\mathbb{R}^n_+} |x|^2 x_n e^{-|x|^2}\, dx\right]\\
&=(1-p)\left[\frac{\frac12m|v|^2}{\kappa T_{\overline{i}}}-\frac{1}{\pi^{\frac{n-1}{2}}}\int_{\mathbb{R}^{n-1}}\left(1+ |x|^2\right)  e^{-|x|^2}\, dx\right]\\
&=(1-p)\left[\frac{\frac12m|v|^2}{\kappa T_{\overline{i}}}-\frac{n+1}{2}\right],
\end{aligned}
\end{equation*}
and
\begin{equation*}
g(i,v)= p\int_{\mathcal{W}^+_{\overline{i}}}\frac{E_0(v)-E_0(u)}{\kappa T_{\overline{i}}}\, d\mu_{\overline{i}}(u)
=p\left[\frac{\frac12 m|v|^2}{\kappa T_{\overline{i}}}-\frac{n+1}{2}\right].
\end{equation*}
These equations are natural having in mind energy equipartition.   
 (One should bear in mind the difference
between the energy distribution on the surface and in bulk. For the latter, one has $n$ times  $\frac12 \kappa T_{\overline{i}}$ whereas,
for the former, the  normal component of the velocity contributes twice as the other components.)
By Proposition \ref{total_residence_entropy}, the expected total entropy produced during the residence time in the compartment is
\begin{equation*}
\mathbb{E}_{(i,v)}[S]=\left[(I-Q)^{-1} (h+g)\right](i,v).
\end{equation*}
We use again Equation (\ref{resolved_applied_to_g}), keeping in mind that $
\mu_i\left(E_0\right)=\frac{n+1}{2}\kappa T_i$.
 (In  (\ref{resolved_applied_to_g})  $g$ is a general function). 
We obtain
\begin{equation*}
\mu_j(h+g)=\frac{n+1}{2}\left(\frac{T_j}{T_{\overline{j}}} -1\right)
\end{equation*}
and 
\begin{equation*}
\mathbb{E}_{(i,v)}[S]=\frac{E_0(v)}{\kappa T_{\overline{i}}}-\frac{n+1}{2} +\frac{1-p}{p(2-p)}\frac{n+1}{2}\left[\frac{T_{\overline{i}}}{T_i} +(1-p)\frac{T_i}{T_{\overline{i}}} -(2-p)\right].
\end{equation*}
The second term in this expression vanishes when the two temperatures are equal. In this case, the first term
accounts for the entropy produced as the possibly  out-of-equilibrium initial velocity relaxes towards equilibrium in the isothermal case. 
Also observe that the expected entropy produced per collision obtained by dividing the just obtained expected total entropy  by the expected number of
collisions, $N(i,v)=1/p$, in the limit as $p$ approaches $0$    (closed compartment limit) is
\begin{equation}\label{entropy_production_per_collision_open}
\lim_{p\rightarrow 0}\frac{\mathbb{E}_{(i,v)}[S]}{\mathbb{E}_{(i,v)}[N]}= \frac{n+1}{4}\frac{(T_1-T_2)^2}{T_1T_2}.
\end{equation}

 \subsubsection{Effect of accommodation coefficients in a closed compartment}
 We consider now the single compartment system of Figure \ref{one_compartment_simple}, but suppose $p_1=p_2=0$ (a closed compartment)
 and  $0<\alpha:=\alpha_1=\alpha_2 \leq 1$. The stationary measure $\nu^\text{o}=(\nu_1^\text{o},\nu_2^\text{o})$ satisfies
 the system
 \begin{equation}\label{equations_nu_one_compartment}
 \nu_{{i}}^\text{o}=\nu_{\overline{i}}^\text{o} S_{i}^{+-}, \ \ i=1, 2.
 \end{equation}
 where 
 \begin{equation}\label{example_S+-i}
 \left(S_i^{+-}f\right)(v)=\alpha \mu_i(f) + (1-\alpha)f(\overline{v})
 \end{equation}
 on a test function $f$.
Recall that $\overline{v}$ is the mirror reflection of $v$. By Proposition \ref{proposition_balance} (or as an immediate consequence of Equations (\ref{equations_nu_one_compartment}) and the explicit form of $S_i^{+-}$),
\begin{equation*}
r_1^\text{o}:= \nu_1^\text{o}(W_1^-)=\nu_2^{\text{o}}(W_2^-)=: r_2^\text{o}.
\end{equation*}
So $r_i^\text{o}=1/2$ and 
\begin{equation}\label{equations_nu_o}
\begin{aligned}
\nu_1^{\text{o}} &=\frac12\alpha \mu_1 +(1-\alpha)\text{ref}_*\nu_2^\text{o}\\
\nu_2^{\text{o}} &=\frac12 \alpha \mu_2 +(1-\alpha)\text{ref}_*\nu_1^\text{o}.
\end{aligned}
\end{equation}
Writing $\mu=(\mu_1,\mu_2)$ and 
\begin{equation*}
A:=\frac{\alpha}2\left(\begin{array}{cc}I & 0 \\ 0 & I\end{array}\right), \ \ B:=(1-\alpha)\left(\begin{array}{cc}0 & \text{ref} \\ \text{ref} & 0\end{array}\right),
\end{equation*}
where $\text{ref}$ indicates wall reflection, then Equations (\ref{equations_nu_o}) may be written as
$
\nu^{\text{o}} = \mu A +\nu^{\text{o}} B.
$
This is easily solved for $\nu^\text{o}$:
\begin{equation*}
\nu^{\text{o}} = \mu A (I-B)^{-1}=\frac{1}{2(2-\alpha)}\mu \left(\begin{array}{cc} I & (1-\alpha)\text{ref} \\ (1-\alpha)\text{ref} & I\end{array}\right).
\end{equation*}
Using the same notation $\mu_i$ for the Maxwellian on either side of the wall, $\text{ref}_*\mu_i=\mu_i$, we obtain
\begin{equation}\label{expression_nuo1_nuo2}
(\nu^\text{o}_1,\nu^\text{o}_2)=\left( \frac{\mu_1+(1-\alpha)\mu_2}{2(2-\alpha)},\frac{(1-\alpha)\mu_1+\mu_2}{2(2-\alpha)}\right).
\end{equation} 
Naturally, as $\alpha$ approaches $1$, the stationary probability measure approaches $\nu^\text{o}_i=\frac12\mu_i$.

The entropy production rate for this system may be calculated from the equation in  Corollary \ref{corollary_of_main_theorem}. Noting  that $\nu^\text{i}_j=\nu^{\text{o}}_{\overline{j}}$, that equation is
\begin{equation*}
\begin{aligned}
e_p&=\sum_{j=1, 2}\int_{W_{{j}}^+\times W_j^-} \frac{E_0(u)-E_0(v)}{\kappa T_j} d\nu^{\text{o}}_{\overline{j}}(u) S_j^{+-}(dv|u)\\
&=\sum_{j=1,2}\frac1{\kappa T_j}\int_{W^+_j}\left[\int_{W^-_{{j}}}\left(E_0(u)-E_0(v)\right) S_j^{+-}(dv|u)\right]\, d\nu^{\text{o}}_{\overline{j}}(u)\\
&=\sum_{j=1,2}\frac1{\kappa T_j}\int_{W^+_j}\left[E_0(u)-S_j^{+-}(E_0|u)\right]\, d\nu_{\overline{j}}^\text{o}(u).
\end{aligned}
\end{equation*}
Using Equation (\ref{example_S+-i}),
\begin{equation*}
S_j^{+-}(E_0|u)=\alpha \mu_j(E_0) + (1-\alpha)E_0(u).
\end{equation*}
Using (see Equation (\ref{expression_nuo1_nuo2})) that $\nu_j^\text{o}=\frac{\mu_j+(1-\alpha)\mu_{\overline{j}}}{2(2-\alpha)}$, 
\begin{equation*}
\begin{aligned}
e_p&=\sum_{j=1,2}\frac\alpha{\kappa T_j}\int_{W_j^+}\left(E_0(u)-\mu_j(E_0)\right)\, d\nu^\text{o}_{\overline{j}}(u)\\
&= \sum_{j=1,2}\frac\alpha{\kappa T_j}\left(\nu_{\overline{j}}^{\text{o}}(E_0)-\frac12\mu_j(E_0)\right)\\
&= \sum_{j=1,2}\frac\alpha{2(2-\alpha)}\frac{\mu_{\overline{j}}(E_0)-\mu_j(E_0)}{\kappa T_j}\\
&=\frac{\alpha}{2(2-\alpha)}\left(\mu_{2}(E_0)-\mu_1(E_0)\right)\left(\frac1{\kappa T_1}-\frac1{\kappa T_2}\right).
\end{aligned}
\end{equation*}
We already know (from previous calculation) that $\mu_j(E_0)=\frac{n+1}{2}\kappa T_j$. Thus
\begin{equation*}
e_p= \frac{\alpha}{2-\alpha}\frac{n+1}{4}\left(T_2-T_1\right)\left(\frac1{T_1}-\frac1{T_2}\right)=\frac{\alpha}{2-\alpha}\frac{n+1}{4}\frac{\left(T_2-T_1\right)^2}{T_1T_2}.
\end{equation*}
This should be compared with Equation (\ref{entropy_production_per_collision_open}) in the limit of $\alpha$ approaching $1$. Naturally,
entropy production decreases as $\alpha$ approaches $0$.
\subsubsection{Effect of potential $C$}\label{sec:potentialC}
This example is another special case of the configuration shown in  Figure \ref{one_compartment_simple}. It will be used in the two and three compartment examples below. Wall $W_1$ has temperature $T$, accommodation coefficient $\alpha=1$, 
transmission coefficient  $p$, and potential $C=0$. Wall $W_2$ is a pure potential barrier with potential parameter $C$  ($p=1$ and  no assignment of $T$ and $\alpha$).
The random operators are 
\begin{equation*}
\begin{aligned}
S_1^{--}(f|v)&= p\mu(f), & S_2^{--}(f|v)&= \mathbbm{1}_{\mathsf{T}}(v)f(v) \\
S_1^{-+}(f|v)&=(1-p)\mu(f), & S_2^{-+}(f|v)&=\mathbbm{1}_{\mathsf{R}}(v)f(\overline{v})\\
S_1^{+-}(f|v)&=(1-p)\mu(f), & S_2^{+-}(f|v)&=\mathbbm{1}_{\mathsf{R}}(v)f(\overline{v})\\
S_1^{++}(f|v)&=p\mu(f), & S_2^{++}(f|v)&=\mathbbm{1}_{\mathsf{T}}(v)f(v).
\end{aligned}
\end{equation*}
Also recall the expression for $Q$ given in Equation (\ref{Q_equation}). Setting $f_{i}=f|_{W_i^-}$,  
and identifying $W_i^-=W_{\overline{i}}^+$, we have
\begin{equation*}
\left[(I-Q)f\right](i,v)=\begin{cases}
f_{1}(v)-\mathbbm{1}_{\mathsf{R}}(v) f_{2}(\overline{v}) & \text{ for }  i=1 \text{ and } v\in W_1^-,\\
f_{2}(v)-(1-p)\mu(f_{1})  & \text{ for } i=2  \text{ and }  v\in W_2^-. 
\end{cases}
\end{equation*}
Here $\mu=\mu_T$ is the Maxwellian at temperature $T$.  The resolvent is
\begin{equation*}
[(I-Q)^{-1}\varphi](i,v)=\varphi_i(v) + \begin{cases}
\mathbbm{1}_{R}(v)\left[\varphi_2(\overline{v}) + \frac{(1-p)(\mu(\varphi_1)+\mu(\mathbbm{1}_{\mathsf{R}}\varphi_2))}{1-(1-p)\mu(\mathsf{R})}\right]& \text{ if } i=1,\\[2ex]
  \frac{(1-p)(\mu(\varphi_1)+\mu(\mathbbm{1}_{\mathsf{R}}\varphi_2))}{1-(1-p)\mu(\mathsf{R})} & \text{ if } i=2.
\end{cases}
\end{equation*}
(For convenience of notation, the reflection of  $\mathsf{R}$ is denoted by this same symbol.)  An elementary calculation gives
\begin{equation*}
\mu(\mathsf{T})=\exp\left(-\frac{C}{\kappa T}\right),
\end{equation*}
and $\mu(\mathsf{R})=1-\mu(\mathsf{T})$. The expected number of collisions for the initial condition $(i,v)$ is (setting $\varphi=1$)
\begin{equation*}
\mathbb{E}_{(i,v)}[N]=1 + \begin{cases}
 \frac{2-p}{1-(1-p)\mu(\mathsf{R})}\mathbbm{1}_{\mathsf{R}}(v)& \text{ if } i=1,\\[2ex]
\frac{(1-p)(1+\mu(\mathsf{R}))}{1-(1-p)\mu(\mathsf{R})} & \text{ if } i=2.
\end{cases}
\end{equation*}

The function $\varphi(i,v):=(h+g)(i,v)$ (defined in  (\ref{definition_h})) is
\begin{equation*}
\varphi(i,v)= \begin{cases}
0 & \text{ if } i=1\\
\frac{1}{\kappa T}\left(E_0(v) - \mu(E_0)\right) & \text{ if } i=2.
\end{cases}
\end{equation*}
 Recall that $\mu\left(E_0\right)=\frac{n+1}{2}\kappa T$.
The expected entropy produced during the stay in the compartment is obtained from
\begin{equation*}
\left[(I-Q)^{-1}(\varphi)\right](i,v)=\varphi_i(v) +\begin{cases}
\mathbbm{1}_{\mathsf{R}}(v) \left[\varphi_2(\overline{v})+\frac{(1-p)\mu(\mathbbm{1}_{\mathsf{R}}\varphi_2)}{1-(1-p)\mu(\mathsf{R})}\right]&\text{ if } i=1\\[2ex]
\frac{(1-p)\mu(\mathbbm{1}_{\mathsf{R}}\varphi_2)}{1-(1-p)\mu(\mathsf{R})} & \text{ if } i=2,
\end{cases}
\end{equation*}
where $v\in W_i^-$ is the initial velocity.
Recall that the transmission set $\mathsf{T}$ (respectively, $\mathsf{R}$) is defined as the set of states for which the normal component of the velocity, $v_n$, is greater than (respectively, smaller than)
$\gamma:=\sqrt{2C/m}$. 
To simplify the notation, let us introduce
\begin{equation*}
D:=1-(1-p)\mu(\mathsf{R})=p+(1-p)\exp\left(-\frac{C}{\kappa T}\right)
\end{equation*}
and a simple calculation gives
\begin{equation}\label{equation_A}
A:=\mu(\mathbbm{1}_{\mathsf{R}} E_0)-\mu(\mathsf{R})\mu(E_0)=-C\exp\left(-\frac{C}{\kappa T}\right).
\end{equation}
We can now write  
\begin{equation}\label{entropy_equation_1}
\mathbb{E}_{(1,v)}[S] = \begin{cases} \varphi_2(\overline{v}) + \dfrac{1-p}{\kappa T}\,\dfrac{A}{D}, & \text{if } v_n < \gamma,\\[6pt]
0, & \text{if } v_n > \gamma.
\end{cases}
\end{equation}
and  
\begin{equation}\label{entropy_equation_2}
\mathbb{E}_{(2,v)}[S]=\varphi_2(v) +\frac{1-p}{\kappa T} \frac{A}{D}.
\end{equation}
Note that if $C=0$, then $\gamma=0$, $\mathsf{R}$ is empty, $A=0$, and the extra term $\mathbb{E}_{(i,v)}[S]-\varphi_2(v)$ disappears. The particle always exits through wall $W_2$ on its first arrival, so only the initial thermal relaxation contributes to entropy production.  If $C>0$, the barrier reflects slow particles; after the first thermalization with wall $W_1$, the velocity is drawn from $\mu$. With probability $\mu(\mathsf{T})=e^{-C/\kappa T}$, the particle escapes and, with probability $\mu(\mathsf{R})$, it returns to wall $W_1$ carrying a biased energy distribution. The total entropy is the sum over this geometric series of returns, giving rise to the extra term on the right-hand side of Equations (\ref{entropy_equation_1}) and (\ref{entropy_equation_2}), despite this being an isothermal compartment.

We record here, for use in the next  section, the components of     the compartment scattering operator.  See Equations (\ref{in_out_equation_example}) and (\ref{equation_A12}). 
\begin{equation*}%\label{in_out_equation_example_2}
\begin{aligned}
\nu_1^\text{o}&=\left[(1-p)\nu_1^\text{i}(H^+) +\frac{p\nu_2^{\text{i}}(\mathsf{T}) +p^2\nu_1^\text{i}(H^+)\mu(\mathsf{R})}{D}\right]\mu\\[1ex]
\nu_2^\text{o}&= \text{ref}_*\nu_2^{\text{i}}|_{\mathsf{R}}+\frac{p\nu_1^{\text{i}}(H^+)+(1-p)\nu_2^{\text{i}}(\mathsf{T})}{D}\mu|_{\mathsf{T}}\\[1ex]
\eta_1&=\frac{p\nu_1^{\text{i}}(H^+) + (1-p)\nu_2^{\text{i}}(\mathsf{T})}{D}\mu\\
\eta_2&=\nu_2^{\text{i}}|_{\mathsf{T}} +\frac{p\nu_1^{\text{i}}(H^+) + (1-p)\nu_2^{\text{i}}(\mathsf{T})}{D} \mu|_\mathsf{R}.
\end{aligned}
\end{equation*}
Here we are using the identification of all half-spaces with $H^+$ as well as the  notation $\mu|_{\mathsf{R}}(f)=\mu(\mathbbm{1}_{\mathsf{R}}f).$

We also record for later use the expected sojourn time in the compartment, which is obtained from the resolvent applied to the function 
\begin{equation*}
\varphi(i,v)=t(i,v)=\frac{\ell}{|v\cdot\mathbf{n}_i|}.
\end{equation*}
A straightforward calculation gives the value
\begin{equation}\label{expected_sojourn_time_C}
\begin{aligned}
\mathbb{E}_{(1,v)}[t]&=\left(1+\mathbbm{1}_{\mathsf{R}}(v)\right)\frac{\ell}{|v\cdot\mathbf{n}_1|} + \mathbbm{1}_{\mathsf{R}}(v)  \frac{1-p}{D}\ell\sqrt{\frac{\pi m}{2\kappa T}}\left[1+\text{erf}\left(\sqrt{\frac{C}{\kappa T}}\right)\right]\\
\mathbb{E}_{(2,v)}[t]&=\frac{\ell}{|v\cdot \mathbf{n}_2|} + \frac{1-p}{D}\ell\sqrt{\frac{\pi m}{2\kappa T}}\left[1+\text{erf}\left(\sqrt{\frac{C}{\kappa T}}\right)\right]\,
\end{aligned}
\end{equation}
in which $\text{erf}(x)=\frac2{\sqrt{\pi}}\int_0^x e^{-t^2} dt$ is the error function. 
\subsection{A two-compartment  example} 
In this example of a two-compartment system, compartment  $i\in \{1,2\}$  is bounded by two walls: one wall, $W_i$, is  reflecting ($p=0$), with temperature
$T_i$, accommodation coefficient $\alpha=1$, and zero potential value; the other, $W_0$, shared by the two compartments, is a pure potential barrier with potential $C\geq 0$. See Figure \ref{two_compartments_new}.

To find the entropy production rate for this system we use Theorem \ref{entropy_multicompartment} and the single compartment results from the previous section. The expected number of wall collisions and the expected total entropy produced
 during the whole sojourn in compartment $i$ are
\begin{equation*}
\mathbb{E}_{(i,v)}[N] =2 \exp\left({\frac{C}{\kappa T_i}}\right), \ \ \mathbb{E}_{(i,v)}[S]=\frac{E_0(v)-\frac{n+1}{2}\kappa T_i-C}{\kappa T_i}.
\end{equation*}
The expected sojourn time is obtained from  Equation (\ref{expected_sojourn_time_C}) by setting $p=0$.
These expected quantities must be computed using the stationary exit distribution from a compartment and into the other. 
Since the thermal wall is fully accommodating ($\alpha_i=1$), the outgoing velocity when the particle finally 
escapes through the potential barrier is distributed as the Maxwellian $\mu_i$ conditioned on the transmission set $\mathsf{T}$. Thus
\begin{equation*}
\nu_{\text{exit},i} =\frac{1}{\mu_i(\mathsf{T})} \mu_i|_{\mathsf{T}}, \ \ \mu_i(\mathsf{T})=\exp\left(-\frac{C}{\kappa T_i}\right).
\end{equation*}

\begin{figure}[ht]
\begin{center}
\includegraphics[width=5in]{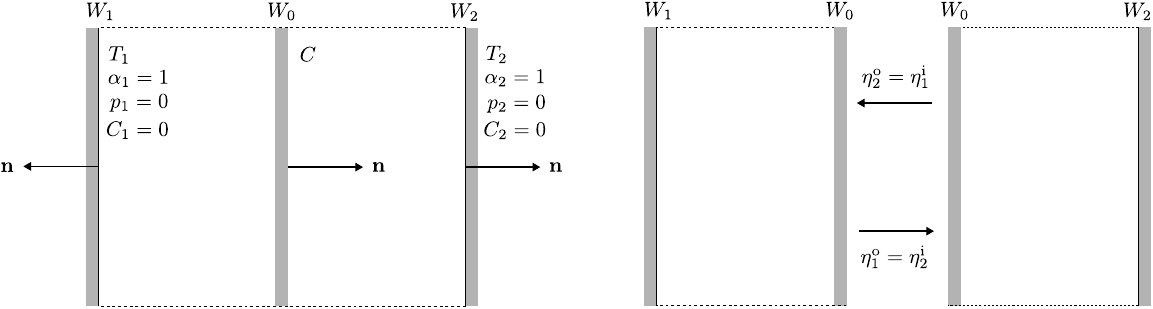}\ \ 
\caption{{\small  We analyze the closed two-compartment system on the left by decomposing it into two communicating
open compartments, using results already obtained for single compartment systems.}}
\label{two_compartments_new}
\end{center}
\end{figure} 

We also need the stationary entrance measure $\eta$. Since the particle trajectory alternates deterministically between the two compartments, the stationary weight of each compartment is $\frac12$, and the velocity distribution when entering compartment $j$ is the
exit distribution of the other compartment. Thus
\begin{equation*}
d\eta^{\text{i}}_j(v)=\frac12 \frac{1}{\mu_{\overline{j}}(\mathsf{T})} \mathbbm{1}_\mathsf{T}(v) \, d\mu_{\overline{j}}(v), \ \ j=1,2.
\end{equation*}
(Recall the notation: $\overline{1}=2$, $\overline{2}=1$.) Also note that
\begin{equation*}
\int_\mathsf{T} E_0(v) \, d\mu_{\overline{j}}(v)= \mu_{\overline{j}}(\mathsf{T}) \frac{n+1}{2}\kappa T_{\overline{j}} + C\exp\left(-\frac{C}{\kappa T_{\overline{j}}}\right),
\end{equation*}
obtained from  Equation (\ref{equation_A}), and
\begin{equation*}
\int_{\mathsf{T}}\mathbb{E}_{(j,v)}[S]\, d\mu_{\overline{j}}(v)=\mu_{\overline{j}}(\mathsf{T}) \frac{n+1}{2}\frac{T_{\overline{j}}-T_j}{T_j}.
\end{equation*}
Therefore,
according to Theorem \ref{entropy_multicompartment},
\begin{equation*}
e_p = \frac{\sum_{j=1,2}\int_{\mathsf{T}}\mathbb{E}_{(j,v)}[S]\, d\eta_j^{\text{i}}(v)}{\sum_{j=1,2}\int_{\mathsf{T}}\mathbb{E}_{(j,v)}[N]\, d\eta_j^{\text{i}}(v)}=\frac{n+1}{4} \frac{(T_2-T_1)^2}{T_1T_2} \frac{1}{\exp\left(\frac{C}{\kappa T_1}\right)+\exp\left(\frac{C}{\kappa T_2}\right)}.
\end{equation*}
Note that as $C$ tends to $0$, we obtain $e_p=\frac{n+1}{8} \frac{(T_2-T_1)^2}{T_1T_2}$, which is half the value shown in Equation (\ref{entropy_production_per_collision_open}). The reason for the difference is that collisions with the now invisible middle wall are being counted.  Therefore it is more meaningful to look at entropy production per unit time, $\dot{e}_p$. This is easily obtained from results already provided:
\begin{equation*}
\dot{e}_p=\frac{n+1}{2\ell}\sqrt{\frac{\kappa}{2\pi m}} \frac{(T_2-T_1)^2}{T_1T_2}
\left[
\frac{\exp\left(\frac{C}{\kappa T_1}
\right)}{\sqrt{T_1}}
+\frac{\exp\left(\frac{C}{\kappa T_2}\right)}{\sqrt{T_2}}
\right]^{-1}.
\end{equation*}
Note that the expected time spent in compartment $j$ during a sojourn, averaged over the stationary entrance distribution,  and the expected entropy produced per visit, are
\begin{equation*}
T_{\text{average}}=\frac{\ell}{2} \sqrt{\frac{2\pi m}{\kappa}}\left[
\frac{\exp\left(\frac{C}{\kappa T_1}
\right)}{\sqrt{T_1}}
+\frac{\exp\left(\frac{C}{\kappa T_2}\right)}{\sqrt{T_2}}
\right], \ \ S_{\text{average}}=\frac{n+1}{4}\frac{(T_2-T_1)^2}{T_1T_2}.
\end{equation*}

\subsection{A three-compartment  example} 
Our last example is a system consisting of three compartments separated by three walls in cyclic fashion. See Figure \ref{three_compartments_new}.
 \begin{figure}[ht]
\begin{center}
\includegraphics[width=3in]{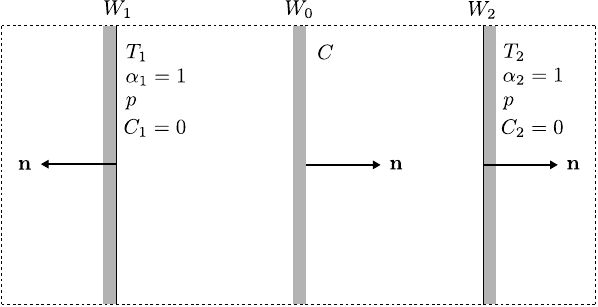}\ \ 
\caption{{\small  In this three-compartment system, the far right and far left parts are connected by a periodic boundary and  form the third compartment of a cyclic configuration. Different from the example of the previous subsection, $p=p_1=p_2$ may be positive.}}
\label{three_compartments_new}
\end{center}
\end{figure} 
Two compartments share  the potential barrier wall $W_0$ and a third is bounded by $W_1$ and $W_2$. 
In this example, entropy production is typically  accompanied  by a  probability flux and circulation.

For the two types of wall we have the following random scattering operators:
\begin{enumerate}
\item Thermal wall with parameters $T_i$, $\alpha$, $p$, and $C=0$:
\begin{equation*}
\begin{aligned}
S_i^{\epsilon \epsilon}(f|v)&=p \left[\alpha \mu_i(f) +(1-\alpha)f(v)\right]\\
S_i^{\epsilon \overline{\epsilon}}(f|v)&=(1-p)\left[\alpha \mu_i(f) +(1-\alpha)f(\overline{v})\right].
\end{aligned}
\end{equation*}
Here $\epsilon\in \{+, -\}$.
\item Potential barrier wall with potential parameter $C$:
\begin{equation*}
\begin{aligned}
S^{\epsilon \epsilon}(f|v)&=\mathbbm{1}_{\mathsf{T}}(v) f(v)\\
S^{\epsilon \overline{\epsilon}}(f|v)&=\mathbbm{1}_{\mathsf{R}}(v) f(\overline{v}).
\end{aligned}
\end{equation*}
\end{enumerate}

\subsubsection{The entrance chain}\label{sec:stationary}
Recall the entrance Markov chain introduced in Section \ref{entropy_composition_section}. Its transition kernel, in the present example,
is given as follows:
\begin{equation*}
\begin{aligned}
p_{ij}(A|(k,w))&:= \text{probability of entering compartment $j$ with  velocity in   $A\subseteq \mathbb{R}^n_+$}\\
                      &\ \ \ \  \ \text{given that the particle starts from wall $k$ of compartment $i$}\\
                       &\ \ \ \ \ \text{with velocity $w$.}
\end{aligned}
\end{equation*}
See Figure \ref{exploded}. We identify all wall subspaces with $\mathbb{R}^n_+$. The stationary probability measure  for
the entrance Markov chain will be denoted
\begin{equation*}
\begin{aligned}
\eta_{i}(j,A)&:=\text{probability that the particle enters compartment $i$ from wall $j$}\\
& \ \ \ \ \ \text{with velocity in $A$.}\\
\end{aligned}
\end{equation*}

\begin{figure}[ht]
\begin{center}
\includegraphics[width=2.5in]{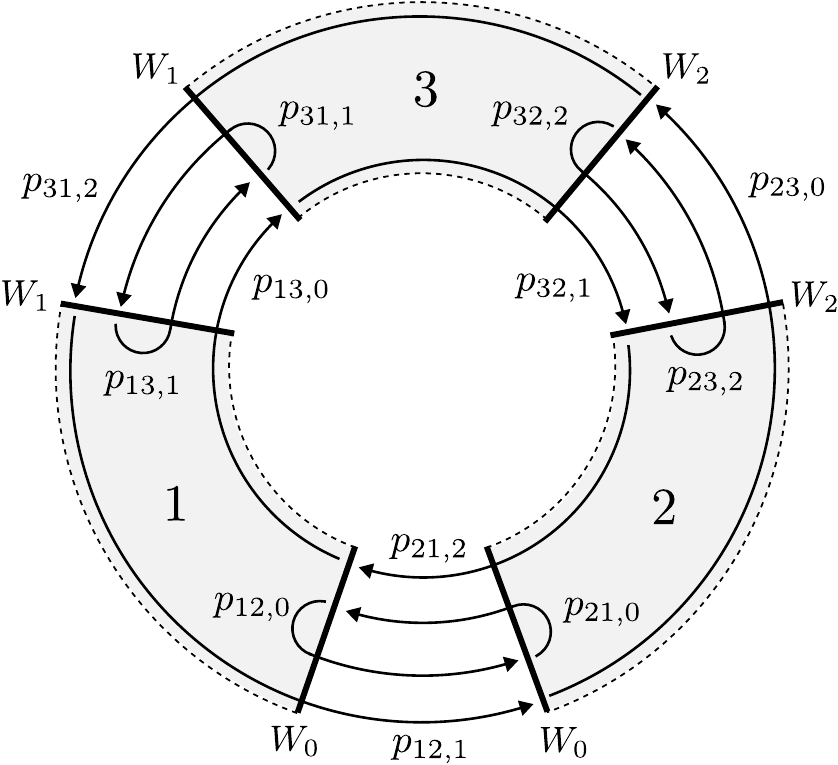}\ \ 
\caption{{\small Transition probabilities of the entrance Markov chain: $p_{ij,k}:=p_{ij}(dv|(k,w))$ is the conditional distribution of entering
compartment $j$ from compartment $i$ given that the particle started from wall $k$ of compartment $i$ with velocity $w$.}}
\label{exploded}
\end{center}
\end{figure} 
 
Define
\begin{equation*}
q_i:=\mu_i(\mathsf{T})=\exp\left(-\frac{C}{\kappa T_i}\right), \ \ a_i:= \frac{p}{1-(1-p)(1-q_i)}, \ \ b_i:=\frac{1-p}{1-(1-p)(1-q_i)},
\end{equation*}
for $i=1,2$, and
\begin{equation*}
a:=\frac{p}{1-(1-p)^2}=\frac1{2-p}, \ \ b:=\frac{(1-p)p}{1-(1-p)^2}=\frac{1-p}{2-p}.
\end{equation*}
  Then the transition  kernel for this example is  
\begin{equation*}
\begin{aligned}
p_{31}(A|(i,w))&=\begin{cases}
b\mu_1(A) &\text{ if } i=1\\[1ex]
a\mu_1(A) &\text{ if } i=2
\end{cases}\\
p_{32}(A|(i,w))&=\begin{cases}
a\mu_2(A) &\text{ if } i=1\\[1ex]
b\mu_2(A) &\text{ if } i=2
\end{cases}\\
p_{13}(A|(i,w))&=\begin{cases}
a_1\mathbbm{1}_{\mathsf{R}}(w)\mu_1(A) &\text{ if } i=1\\[1ex]
a_1 \mu_1(A)  &\text{ if } i=0
\end{cases}\\
p_{23}(A|(i,w))&=\begin{cases}
a_2\mathbbm{1}_{\mathsf{R}}(w)\mu_2(A) &\text{ if } i=2\\[1ex]
a_2 \mu_2(A)  &\text{ if } i=0
\end{cases}\\
p_{12}(A|(i,w))&=\begin{cases}
\mathbbm{1}_{\mathsf{T}\cap A}(w)  +\mathbbm{1}_{\mathsf{R}}(w) {b_1}\mu_1(\mathsf{T}\cap A)&\text{ if } i=1\\[1ex]
{b_1} \mu_1(\mathsf{T}\cap A) &\text{ if } i=0
\end{cases}\\
p_{21}(A|(i,w))&=\begin{cases}
\mathbbm{1}_{\mathsf{T}\cap A}(w)  +\mathbbm{1}_{\mathsf{R}}(w){b_2}\mu_2({\mathsf{T}}\cap A)&\text{ if } i=2\\[1ex]
{b_2} \mu_2(\mathsf{T}\cap A) &\text{ if } i=0
\end{cases}
\end{aligned}
\end{equation*}
for $A\subseteq \mathbb{R}^n_+$.
These values may be obtained by an elementary geometric sum argument. For example, for $p_{12}(A|(1,w))$
assume  the particle starts  in compartment $1$, wall $1$, with velocity $w$, and
 add the probabilities of it crossing wall $W_0$ into compartment $2$, with velocity in $A$, in
\begin{equation*}
\begin{aligned}
\text{ $1$ step: } &\text{ $\mathbbm{1}_{\mathsf{T}\cap A}(w)$;}\\
\text{ $3$ steps: } & \text{ $\mathbbm{1}_{\mathsf{R}}(w) (1-p) \mu_1(\mathsf{T}\cap A)$ (full thermalization on return to $W_1$);}\\
\text{ $5$ steps: } & \text{ $\mathbbm{1}_{\mathsf{R}}(w) (1-p)^2(1-q_1)\mu_1(\mathsf{T}\cap A)$;}\\
\text{ $7$ steps: } & \text{ $\mathbbm{1}_{\mathsf{R}}(w) (1-p)^3(1-q_1)^2\mu_1(\mathsf{T}\cap A)$;}\\
\cdots \ \ \ \ & \ \ \ \   \cdots\\
\text{ sum: } & \mathbbm{1}_{\mathsf{T}\cap A}(w)+ \mathbbm{1}_{\mathsf{R}}(w) (1-p) \left[ 1+ (1-p)(1-q_1) \right.\\
&\ \ \ \ \ \ \ \  \ \ \ \ \ \ \ \ \ \ \ \ \ \ \ \ \ \  \ \ \ \ \  \ \ \ \ \ \ \ \ \ \ \ \ \ \ \ \ \  \ \ \  \left.+ (1-p)^2(1-q_1)^2+\cdots\right] \mu_1(\mathsf{T}\cap A).
\end{aligned}
\end{equation*}
At odd step numbers $2m+1>1$, the particle arrives at $W_0$ fully thermalized with  $W_1$, and it  crosses $W_0$ (with its velocity unchanged) if and only if   the velocity
lies in $\mathsf{T}$. $W_0$ is invisible to the particle on $\mathsf{T}$
and  specularly reflecting on $\mathsf{R}$. Thus the post-crossing velocity distribution is $q_1^{-1}\mu_1|_\mathsf{T}$.  Note the cancellation of $q_1$ in \begin{equation*}
\mathbbm{1}_{\mathsf{R}}(w) (1-p)^{m}(1-q_1)^{m-1}q_1 q_1^{-1}d\mu_1|_{\mathsf{T}}(v),
\end{equation*} in which $(1-p)^{m}(1-q_1)^{m-1}q_1$ is the probability of $m$ diffuse reflections on $W_1$ times  $m$ reflections on $W_0$ times the probability of finally crossing the barrier at collision step $2m+1$.

We now need the stationary probability  measure for the entrance chain. The component  of this measure associated with compartment $i$, wall $j$, satisfies the homogeneous equation
\begin{equation*}
\eta_i(j, A)=\int_{\mathbb{R}^n_+}\eta_{k}(j,dv) p_{ki}(A|(j,v))+\int_{\mathbb{R}^n_+}\eta_{k}(j',dv) p_{ki}(A|(j',v)),
\end{equation*}
where $k$ is the compartment sharing with $i$ wall $W_j$, $j'$ labels the wall of compartment $k$ opposite to $W_j$, and   $A\subset\mathbb{R}^n_+$.  Thus we have six homogeneous equations (here we write $\mathsf{H}=\mathbb{R}^n_+$ for simplicity):
\begin{equation*}
\begin{aligned}
\eta_1(1,A)&=\left[b\eta_3(1,\mathsf{H})+ a\eta_3(2,\mathsf{H})\right]\mu_1(A)\\
\eta_2(2,A)&=\left[a\eta_3(1,\mathsf{H})+ b\eta_3(2,\mathsf{H})\right]\mu_2(A)\\
\eta_1(0,A)&=\eta_2(2,\mathsf{T}\cap A) +{b_2}\left[\eta_2(0,\mathsf{H})+\eta_2(2,\mathsf{R})\right]\mu_2(\mathsf{T}\cap A)\\
\eta_2(0,A)&=\eta_1(1,\mathsf{T}\cap A) +b_1\left[\eta_1(0,\mathsf{H})+\eta_1(1,\mathsf{R})\right]\mu_1(\mathsf{T}\cap A)\\
\eta_3(1,A)&=a_1\left[\eta_1(0,\mathsf{H})+ \eta_1(1,\mathsf{R})\right]\mu_1(A)\\
\eta_3(2,A)&=a_2\left[\eta_2(0,\mathsf{H})+ \eta_2(2,\mathsf{R})\right]\mu_2(A).
\end{aligned}
\end{equation*}
Defining the coefficients  $\eta_{ij}:=\eta_i(j,\mathsf{H})$, we may rewrite these equations as 
\begin{equation*}
\begin{aligned}
\eta_1(1,A)&=\eta_{11}\mu_1(A)\\
\eta_2(2,A)&=\eta_{22}\mu_2(A)\\
\eta_1(0,A)&=\left[(1+b_2(1-q_2))\eta_{22}+b_2\eta_{20}\right]\mu_2(\mathsf{T}\cap A)\\
\eta_2(0,A)&=\left[(1+b_1(1-q_1))\eta_{11}+b_1\eta_{10}\right]\mu_1(\mathsf{T}\cap A)\\
\eta_3(1,A)&=a_1\left[\eta_{10} +(1-q_1)\eta_{11}\right]\mu_1(A)\\
\eta_3(2,A)&=a_2\left[\eta_{20} +(1-q_2)\eta_{22}\right]\mu_2(A).
\end{aligned}
\end{equation*}
Using the previously defined $D_i$,
\begin{equation*}
D_i:=1-(1-p)(1-q_i)=\left[1+b_i(1-q_i)\right]^{-1},
\end{equation*}
we can write the   homogeneous system of linear equations for the $\eta_{ij}$ as
\begin{equation}\label{system_eta}
\begin{aligned}
\eta_{11}&=b\eta_{31}+a\eta_{32}\\
\eta_{22}&=a\eta_{31}+b\eta_{32}\\
%\eta_{10}&=q_2\left[1+b_2(1-q_2)\right]\eta_{22}+b_2q_2\eta_{20}\\
\eta_{10}&=q_2D_2^{-1}\eta_{22}+b_2q_2\eta_{20}\\
%\eta_{20}&=q_1\left[1+b_1(1-q_1)\right]\eta_{11}+b_1q_1\eta_{10}\\
\eta_{20}&=q_1D_1^{-1}\eta_{11}+b_1q_1\eta_{10}\\
\eta_{31}&=a_1\left[\eta_{10} +(1-q_1)\eta_{11}\right]\\
\eta_{32}&=a_2\left[\eta_{20} +(1-q_2)\eta_{22}\right].
\end{aligned}
\end{equation}
To these we add the normalization condition
\begin{equation*}
\eta_{11}+\eta_{22}+\eta_{10}+\eta_{20}+\eta_{31}+\eta_{32}=1.
\end{equation*}
System (\ref{system_eta}) has the following solution, in which the common denominator $d$ is 
\begin{equation*}
d:=4p^2 + p(5-6p)(q_1+q_2) + (4-6p)q_1q_2,
\end{equation*}
with $0<p\leq 1$ and $0<q_i<1$:
\begin{equation*}
\begin{aligned}
\eta_{10}&= \frac{q_2\left[2(1-p) q_1 + p(1-q_1)\right]}{d}, & \eta_{20}&= \frac{q_1\left[2(1-p) q_2 + p(1-q_2)\right]}{d}   \\
\eta_{31}&= \frac{p\left[2(1-p) q_2 + p(1-q_1)\right]}{d},  & \eta_{32}&= \frac{p\left[2(1-p) q_1 + p(1-q_2)\right]}{d}  \\
\eta_{11}&= \frac{p\left[(1-p) q_1 + (1-2p)q_2+p\right]}{d},  & \eta_{22}&= \frac{p\left[(1-p) q_2 + (1-2p)q_1+p\right]}{d}.
\end{aligned}
\end{equation*}
Let us register here the special case having $p=\frac12$. In this case, $d=(1+q_1)(1+q_2)$ and
\begin{equation}\label{eq:etahalf}
\eta_{i0} = \frac{q_{\overline{i}}}{2(1+q_{\overline{i}})}, \ \ \eta_{3i}=\frac{1-q_i +2q_{\overline{i}}}{4(1+q_1)(1+q_2)}, \ \ 
\eta_{ii} = \frac{1}{4(1+q_{\overline{i}})},\ \ i=1,2.
\end{equation}

\subsubsection{Circulation} Using  Definition  \ref{probability_flux_definition} and Equation (\ref{probability_flux_balance_cycle}),
we can write the probability flux circulation for this example.
First note that
\begin{equation*}
\begin{aligned}
\Phi_{12} &=\int_{\mathsf{H}} \eta_1(1,dv)p_{12}(\mathsf{H}|(1,v)) + \int_{\mathsf{H}} \eta_1(0,dv)p_{12}(\mathsf{H}|(0,v))=\frac{q_1}{D_1}\eta_{11} + b_1 q_1 \eta_{10}\\
&=\eta_{20},
\end{aligned}
\end{equation*}
and 
\begin{equation*}
\begin{aligned}
\Phi_{21} &=\int_{\mathsf{H}} \eta_2(2,dv)p_{21}(\mathsf{H}|(2,v)) + \int_{\mathsf{H}} \eta_2(0,dv)p_{21}(\mathsf{H}|(0,v))=\frac{q_2}{D_2}\eta_{22} + b_2 q_2 \eta_{20}\\
 &=\eta_{10}.
\end{aligned}
\end{equation*}
Therefore, 
\begin{equation*}
J=\Phi_{12} -\Phi_{21}=\eta_{20}-\eta_{10}.
\end{equation*}
Direct verification confirms the equality of the flux computed at the other two walls:
\begin{equation*}
\eta_{20}-\eta_{10}=\eta_{32}-\eta_{22}=\eta_{11}-\eta_{31}=\frac{p(q_1-q_2)}{4p^2 + p(5-6p)(q_1+q_2) + (4-6p)q_1q_2}.
\end{equation*}
Explicitly, for the special value $p=\frac12$,
\begin{equation*}
J=\frac{\frac12\left[\exp\left(-\frac{C}{\kappa T_1}\right) - \exp\left(-\frac{C}{\kappa T_2}\right)\right]}{\left[1+\exp\left(-\frac{C}{\kappa T_1}\right)\right]\left[1+\exp\left(-\frac{C}{\kappa T_2}\right)\right]}.
\end{equation*}
Naturally, the  probability flux vanishes  when $T_1=T_2$. It also vanishes when $C=0$. When $C>0$ and  $T_1>T_2$, the flux is positive (counterclockwise circulation in Figure
\ref{exploded}.) When $T_1<T_2$ the direction of circulation is reversed.   Note that, in compartment $3$, the particle moves from the cold towards the hot wall.
This ratchet effect serves as a simple but potentially useful  model for the phenomenon,  characteristic of the Knudsen gas regime, 
of {\em thermal transpiration} (also known as {\em thermal creep}), in which the molecular mean free path is significantly larger than the dimensions of the container.
It    may be the basis  for the  design of stochastic thermodynamic machines.

\subsubsection{Compartment statistics and entropy production}\label{sec:stats}

The three compartments that appear in the cycle are exactly the two types
analysed in Sections~\ref{sec:openCompartment} and~\ref{sec:potentialC}:
a compartment bounded by one thermal wall and one potential barrier
(compartments~$1$ and~$2$), and a compartment bounded by two thermal
walls (compartment~$3$).  Their sojourn statistics\----expected time and
expected entropy produced\----are obtained by applying the general
formulas of Section~\ref{compartment_statistics} with the resolvent
$(I-Q)^{-1}$ computed in those sections.  We only
quote the results averaged over the appropriate entrance velocity
distributions; the necessary surface Maxwellian integrals are
\begin{equation}\label{eq:integrals}
  g_i:=\mu_i\!\Bigl(\frac1{v_n}\Bigr)=\sqrt{\frac{\pi m}{2\kappa T_i}},\qquad
  \mathrm{erf}_i:=\operatorname{erf}\!\Bigl(\sqrt{\frac{C}{\kappa T_i}}\Bigr),\qquad
  m_i:=\mu_i(E_0)=\frac{n+1}{2}\,\kappa T_i .
\end{equation}

\paragraph{Middle compartment (two thermal walls).}
For a sojourn entered at wall $W_j$ (velocity distributed as $\mu_j$,
cf.\ Section~\ref{sec:stationary}), the expected time and entropy are 
\begin{align}
  \mathbb{E}[t|(3,j)] &=
   \ell_3\,g_j
   +\frac{\ell_3}{p(2-p)}\Bigl[(1-p)\,g_{\bar\jmath}
   +(1-p)^2g_j\Bigr],\label{eq:t3}\\[2pt]
  \mathbb{E}[S|(3,j)] &=
   \frac{m_j-m_{\bar\jmath}}{\kappa T_{\bar\jmath}}
   +\frac{1}{p(2-p)}\Bigl[(1-p)\,\frac{m_{\bar\jmath}-m_j}{\kappa T_j}
   +(1-p)^2\,\frac{m_j-m_{\bar\jmath}}{\kappa T_{\bar\jmath}}\Bigr],\label{eq:S3}
\end{align}
where $\bar\jmath$ denotes the opposite wall.  These expressions follow
from the resolvent of the open-compartment scattering operator in
Section~\ref{sec:openCompartment}, averaged over the entrance Maxwellian.

\paragraph{Barrier compartments (one thermal wall and the barrier).}
Using the resolvent $(I-Q)^{-1}$ obtained in Section~\ref{sec:potentialC}
and the same averaging procedure, the expected sojourn time and entropy
for the two entrance types are found to be
\begin{align}
  \mathbb{E}[t|(i,i)] &=
   \frac{\ell_i\,g_i\,(1+\mathrm{erf}_i)}{D_i},\label{eq:tii}\\[2pt]
  \mathbb{E}[S|(i,i)] &=
   -\,\frac{q_iC}{D_i\,\kappa T_i},\label{eq:Sii}\\[2pt]
  \mathbb{E}[t|(i,0)] &=
   \frac{\ell_i\,g_{\bar\imath}(1-\mathrm{erf}_{\bar\imath})}{q_{\bar\imath}}
   +\frac{(1-p)\,\ell_i\,g_i\,(1+\mathrm{erf}_i)}{D_i},\label{eq:ti0}\\[2pt]
  \mathbb{E}[S|(i,0)] &=
   \frac{m_{\bar\imath}+C-m_i}{\kappa T_i}
   -\,\frac{(1-p)\,q_iC}{D_i\,\kappa T_i},\label{eq:Si0}
\end{align}
with $i=1,2$ and $\bar\imath=3-i$.  
Equations \eqref{eq:tii} and \eqref{eq:Sii} give the expected time and
entropy for an entrance through the thermal wall (the initial velocity is
distributed as $\mu_i$); Equations \eqref{eq:ti0} and \eqref{eq:Si0} give
the corresponding quantities for an entrance through the barrier, where
the initial velocity is distributed as
$q_{\bar{\imath}}^{-1}\mu_{\bar{\imath}}|_{\mathsf{T}}$.
 The
derivation is a straightforward evaluation of $(I-Q)^{-1}\phi$ and
$(I-Q)^{-1}(h+g)$ on the respective initial distributions, using the
integrals $\mu_i(\mathbf{1}_{\mathsf{R}}/v_n)=g_i\,\mathrm{erf}_i$,
$\mu_i(\mathbf{1}_{\mathsf{T}}E_0)=q_i(m_i+C)$ and the identity
$\mu_i(\mathbf{1}_{\mathsf{R}}E_0)-(1-q_i)m_i=-q_iC$.

\paragraph{Averaged statistics.}
With the stationary entrance masses $\eta_{ij}$ from
Section~\ref{sec:stationary} we form the means per transition  
\begin{equation}\label{eq:tSdef}
 \mathbb{E}[t]:=\sum_{(i,j)}\eta_{ij}\,\mathbb{E}[t(i,j)],\qquad
   \mathbb{E}[S]:=\sum_{(i,j)}\eta_{ij}\,\mathbb{E}[S(i,j)] .
\end{equation}

\begin{lemma}[Error function cancellation]\label{lem:erf}
 The coefficient multiplying the term $g_1\mathrm{erf}_1$ in $\mathbb{E}[t]$ equals
$(\ell_1-\ell_2)\,\eta_{20}/q_1$, and that of $g_2\mathrm{erf}_2$ equals
$(\ell_2-\ell_1)\,\eta_{10}/q_2$.  Consequently, if the two barrier
compartments have equal widths ($\ell_1=\ell_2$), all error functions
disappear from $ \mathbb{E}[t]$ regardless of $\ell_3$.
\end{lemma}
\begin{proof}
The terms containing $g_1\mathrm{erf}_1$ come from three places:
\begin{enumerate}
  \item $\mathbb{E}[t|(1,1)]$ in \eqref{eq:tii}: the $\mathrm{erf}_1$
        part is $\displaystyle\frac{\ell_1 g_1\mathrm{erf}_1}{D_1}$;
  \item $\mathbb{E}[t|(1,0)]$ in \eqref{eq:ti0}: the second summand
        contributes $\displaystyle\frac{(1-p)\ell_1 g_1\mathrm{erf}_1}{D_1}$;
  \item $\mathbb{E}[t|(2,0)]$ in \eqref{eq:ti0} (with $i=2$,
        $\bar\imath=1$): the first term is
        $\displaystyle\frac{\ell_2 g_1(1-\mathrm{erf}_1)}{q_1}$, giving a
        $\mathrm{erf}_1$ part of $\displaystyle-\frac{\ell_2 g_1\mathrm{erf}_1}{q_1}$.
\end{enumerate}
Multiplying by the respective stationary masses and summing, the
coefficient of $g_1\mathrm{erf}_1$ in $ \mathbb{E}[t]$ is
\[
  \frac{\ell_1}{D_1}\,\eta_{11}
  +\frac{(1-p)\ell_1}{D_1}\,\eta_{10}
  -\frac{\ell_2}{q_1}\,\eta_{20}.
\]
From the entrance chain Equations (\ref{system_eta}) one has
$\eta_{20}= \frac{q_1}{D_1}\bigl[\eta_{11}+(1-p)\eta_{10}\bigr]$,
hence the first two terms combine to
$\frac{\ell_1}{D_1}\bigl[\eta_{11}+(1-p)\eta_{10}\bigr]
 =\frac{\ell_1}{q_1}\,\eta_{20}$.
Therefore the total coefficient is
$\frac{\ell_1-\ell_2}{q_1}\,\eta_{20}$, as claimed.
The second assertion follows by symmetry.
\end{proof}

For the special case $p=\tfrac12$ and equal widths
$\ell_1=\ell_2=\ell_3=:\ell$, the stationary masses  
\eqref{eq:etahalf} give the compact closed form
\begin{equation}\label{eq:thalf}
  \mathbb{E}[t]
  =\frac{3\ell}{2}\;
   \frac{(1+q_2)\,g_1+(1+q_1)\,g_2}{(1+q_1)(1+q_2)} .
\end{equation}

\paragraph{Clausius form of the entropy production.}
Because every thermal collision resamples the velocity from the local
Maxwellian and barrier events preserve the energy, the mean net heat
absorbed by the walls per transition vanishes, and $\mathbb{E}[S]$
can be written as
\begin{equation}\label{eq:clausius}
 \mathbb{E}[S]
  =Q_e\,\Bigl(\frac1{\kappa T_1}-\frac1{\kappa T_2}\Bigr),
\end{equation}
where $Q_e$ is the mean heat deposited into wall~$1$ per transition.
With \begin{equation*}
\Delta:=m_2-m_1=\tfrac{n+1}2\kappa\,(T_2-T_1)
\end{equation*} one finds
\begin{equation}\label{eq:Qe}
  Q_e=\frac{\Delta}{2p}\,\Bigl[\,p-(2p-1)\bigl(\eta_{31}
  +\eta_{32}\bigr)\Bigr]\;-\;C\,J,
\end{equation}
and in particular at $p=\tfrac12$,
$\displaystyle Q_e=\frac{\Delta}{2}-C\,J$.  The heat current thus
consists of a conduction term (present even without circulation) and a
convective term proportional to the net barrier crossings.

We can also consider entropy production and circulation per unit time:
\begin{equation*} 
  \dot e_p=\frac{\mathbb{E}[S]}{\mathbb{E}[t]},
  \qquad
  J_{\mathrm{time}}=\frac{J}{\mathbb{E}[t]}.
\end{equation*}
 For $p=\tfrac12$ and equal widths,
 we obtain the simpler equations
\begin{equation*}
  \dot e_p
  =\frac{\bigl[(1+q_1)(1+q_2)\,\Delta+(q_2-q_1)\,C\bigr]
   \Bigl(\dfrac1{\kappa T_1}-\dfrac1{\kappa T_2}\Bigr)}
   {3\ell\,\bigl[(1+q_2)\,g_1+(1+q_1)\,g_2\bigr]},
   \end{equation*}
   and  
\begin{equation*}
J_{\mathrm{time}}
  =\frac{q_1-q_2}{3\ell\,\bigl[(1+q_2)\,g_1+(1+q_1)\,g_2\bigr]}.
\end{equation*}
Similar but more complicated equations, giving the explicit dependence of entropy production and circulation on system parameters, can be obtained when thermal accommodation $\alpha$ is not $1$. We do not give those here.

 \end{document}